\documentclass[11pt]{article} 
\usepackage{times}
\usepackage{amsmath,amsfonts}
\usepackage{epsfig,amssymb,amstext,xspace,theorem}
\usepackage{graphicx,color} 
\usepackage{verbatim}

\DeclareMathOperator*{\Exp}{E}
\newcommand{\E}[2][{}]{\ensuremath{{\textstyle \Exp_{#1}}\left[#2\right]}}

\newtheorem{theorem}{Theorem}[section]
\newtheorem{lemma}[theorem]{Lemma}

\newtheorem{claim}[theorem]{Claim}

{\theorembodyfont{\rmfamily} \newtheorem{remark}[theorem]{Remark}}

\def\blksquare{\rule{2mm}{2mm}}
\def\qedsymbol{\blksquare}
\newcommand{\bg}[1]{\medskip\noindent{\bf #1}}
\newcommand{\ed}{{\hfill\qedsymbol}\medskip}
\newenvironment{proof}{\bg{Proof : }}{\ed}
\newenvironment{proofof}[1]{\bg{Proof of #1 : }}{\ed}

\newcommand{\B}{\ensuremath{\mathcal{B}}}
\newcommand{\C}{\ensuremath{\mathcal{C}}}

\newcommand{\F}{\ensuremath{\mathcal F}}
\newcommand{\D}{\ensuremath{\mathcal D}}

\newcommand{\Pc}{\ensuremath{\mathcal P}}
\newcommand{\PP}{\ensuremath{\mathcal P}}

\newcommand{\sm}{\ensuremath{\setminus}}
\newcommand{\es}{\ensuremath{\emptyset}}

\newcommand{\e}{\ensuremath{\epsilon}}

\newcommand{\sse}{\subseteq}

\newcommand{\remove}[1]{}
\newcommand{\ufl}{{\small \textsf{UFL}}\xspace}
\newcommand{\cfl}{{\small \textsf{CFL}}\xspace}
\newcommand{\mfl}{{\small \textsf{MFL}}\xspace}
\newcommand{\mmfl}{\ensuremath{\mathsf{MFL}}}
\newcommand{\shift}{\ensuremath{\mathsf{shift}}}
\newcommand{\spath}{\ensuremath{\mathsf{start}}}
\newcommand{\epath}{\ensuremath{\mathsf{end}}}

\newcommand{\swap}{\ensuremath{\mathrm{swap}}}
\newcommand{\capt}{\ensuremath{\mathsf{cap}}}
\newcommand{\cent}{\ensuremath{\mathsf{center}}}
\newcommand{\cand}{\ensuremath{\mathsf{cand}}}
\newcommand{\next}{\ensuremath{\mathsf{next}}}
\newcommand{\lz}{\ensuremath{p}}
\newcommand{\hi}{\ensuremath{\hat i}}
\newcommand{\al}{\ensuremath{\alpha}}
\newcommand{\sg}{\ensuremath{\sigma}}
\newcommand{\hD}{\ensuremath{\widehat D}}
\newcommand{\hG}{\ensuremath{\widehat G}}

\newcommand{\bon}{\ensuremath{\mathbf{1}}}
\newcommand{\nil}{\ensuremath{\mathsf{nil}}}
\newcommand{\sub}[2][{q}]{\ensuremath{(\ref{#2}_{#1})}}
\newcommand{\blk}{\ensuremath{\B}}
\newcommand{\sib}{\ensuremath{\approx}}
\newcommand{\nsib}{\ensuremath{\not\approx}}

\begin{document}

\title{Local-Search based Approximation Algorithms for Mobile Facility Location Problems%
\footnote{A preliminary version~\cite{AhmadianFS13}, without the results in
Section~\ref{oneswap}, appeared in the Proceedings of the 24th Annual ACM-SIAM Symposium
on Discrete Algorithms, 2013.}}
\author{
         Sara Ahmadian\thanks{{\tt \{sahmadian,zfriggstad,cswamy\}@math.uwaterloo.ca}. 
         Dept. of Combinatorics and Optimization, Univ. Waterloo, Waterloo, ON N2L 3G1.
         Supported in part by NSERC grant 327620-09 and an NSERC Discovery Accelerator
         Supplement award. 
         The second and third authors are also supported by the third author's Ontario
         Early Researcher Award.}   
\and
\addtocounter{footnote}{-1}
         Zachary Friggstad\footnotemark
\and
\addtocounter{footnote}{-1}
         Chaitanya Swamy\footnotemark
}

\date{}

\maketitle

\begin{abstract}
We consider the {\em mobile facility location} (\mfl) problem. 
We are given a set of facilities and clients located in a common metric space 
$G=(V,c)$. 
The goal is to move each facility from its initial location to a destination (in $V$)
and assign each client to the destination of some facility so as to minimize the sum of
the movement-costs of the facilities and the client-assignment costs.
This abstracts facility-location settings where one has the flexibility of moving
facilities from their current locations to other destinations so as to serve clients more 
efficiently by reducing their assignment costs. 

We give the first {\em local-search based} approximation algorithm for this problem and
achieve the best-known approximation guarantee. Our main result is
$(3+\epsilon)$-approximation for this problem for any constant $\epsilon > 0$ using local
search. The previous best guarantee for \mfl was an 8-approximation algorithm due to
~\cite{FriggstadS11} based on LP-rounding. 
Our guarantee {\em matches} the best-known approximation guarantee for the $k$-median
problem. Since there is an approximation-preserving reduction from the $k$-median problem
to \mfl, any improvement of our result would imply an analogous improvement for the
$k$-median problem.  
Furthermore, {\em our analysis is tight} (up to $o(1)$ factors) since the tight example
for the local-search based 3-approximation algorithm for $k$-median can be easily adapted
to show that our local-search algorithm has a tight approximation ratio of 3.
One of the chief novelties of the analysis is that in
order to generate a suitable collection of local-search moves whose resulting inequalities
yield the desired bound on the cost of a local-optimum, we define a tree-like structure
that (loosely speaking) functions as a ``recursion tree'', using which we spawn off
local-search moves by exploring this tree to a constant depth.
Our results extend to the weighted generalization wherein each facility $i$ has a
non-negative weight $w_i$ and the movement cost for $i$ is $w_i$ times the distance
traveled by $i$. 
\end{abstract}

\normalsize

\section{Introduction}
Facility location problems have been widely studied in the Operations Research and
Computer Science communities (see, e.g.,~\cite{MirchandaniF90} and the
survey~\cite{KorteV08}), and have a wide range of applications.
In its simplest version, {\em uncapacitated facility location} (\ufl), we are given a
set of facilities or service-providers with opening costs, and a set of clients that
require service, and we want to open some facilities 
and assign clients to open facilities so as to minimize the sum of the facility-opening
and client-assignment costs.    
An oft-cited prototypical example is that of a company wanting to decide where to locate
its warehouses/distribution centers so as to serve its customers in a cost-effective
manner. 

We consider facility-location problems that abstract settings where facilities are
mobile and may be relocated to destinations near the clients in order to serve them more
efficiently by reducing the client-assignment costs.  
More precisely, we consider the {\em mobile facility location} (\mfl) problem introduced
by~\cite{DemaineHMSOZ09,FriggstadS11}, which generalizes the classical $k$-median 
problem (see below). We are given 
a complete graph $G=(V,E_G)$ with costs $\{c(u,v)\}$ on the edges, a set $\D\subseteq V$ of
clients with each client $j$ having $d_j$ units of demand,  
and a set $\F\subseteq V$ of $k$ initial facility locations. We use the term facility $i$
to denote the facility whose initial location is $i\in\F$.  
A solution $S$ to \mfl moves each facility $i$ to a final location $s_i\in V$ (which could
be the same as $i$), incurring a {\em movement cost} $c(i,s_i)$, 
and assigns each client $j$ to a final location $s\in S$, incurring {\em assignment cost} 
$d_jc(j,s)$. The total cost of $S$ is the sum of all the movement costs and assignment 
costs. 
More formally, noting that each client will be assigned to the location nearest to it in
$S$, we can express the cost of $S$ as 
$$
\mmfl(S):=\sum_{i\in \F}c(i,s_i) + \sum_{j\in\D}d_jc(j,\sigma(j))
$$
where $\sigma(v)$ (for any node $v$) gives the location in $S$ nearest to $v$ (breaking
ties arbitrarily). We assume throughout that the edge costs form a metric. 
We use the terms nodes and locations interchangeably. 

Mobile facility location falls into the genre of {\em movement problems}
introduced by Demaine et al.~\cite{DemaineHMSOZ09}. In these problems, we are given an
initial configuration in a weighted graph specified by placing ``pebbles'' on the nodes
and/or edges; the goal is to move the pebbles so as to obtain a desired final
configuration while minimizing the maximum, or total, pebble movement. 
\mfl was introduced by Demaine et al. as the movement problem where facility- and client-
pebbles are placed respectively at the initial locations of the facilities and 
clients, and in the final configuration every client-pebble should be co-located
\nolinebreak \mbox{with some facility-pebble.}

%\vspace{-1ex}
\paragraph{Our results.}
We give the first {\em local-search based} approximation algorithm for this problem and
achieve the best-known approximation guarantee. Our main result is a
$(3+\epsilon)$-approximation for this problem for any constant $\epsilon > 0$ using a
simple local-search algorithm. This improves upon the previous best 8-approximation
guarantee for \mfl due to Friggstad and Salavatipour~\cite{FriggstadS11}, which is based
on LP-rounding and is not combinatorial.  

The local-search algorithm we consider is quite natural and simple.
Observe that given the final locations of the facilities, we can find the minimum-cost way
of moving facilities from their initial locations to the final locations by solving a
minimum-cost perfect-matching problem (and the client
assignments are determined by the function $\sg$ defined above). Thus, we concentrate
on determining a good set of final locations.  
In our local-search algorithm, at each step, we are allowed to swap in and swap out a
fixed number (say $\lz$) of locations. 
Clearly, for any fixed $\lz$, we can find the best local move efficiently (since the cost
of a set of final locations can be computed in polytime). 
Note that we do not impose any constraints on how the matching between the initial and
final locations may change due to a local move, and a local move might entail moving all
facilities.  
It is important to allow this flexibility, as it is known~\cite{FriggstadS11} that the
local-search procedure that moves, at each step, a constant number of facilities to chosen
destinations has an unbounded approximation ratio. 

Our main contribution is a tight analysis of this local-search algorithm 
(Section~\ref{3apx}). 
Our guarantee {\em matches} (up to $o(1)$ terms) the best-known approximation guarantee
for the $k$-median problem.  
Since there is an approximation-preserving reduction from the $k$-median problem to
\mfl~\cite{FriggstadS11}---choose arbitrary initial facility locations and give each
client a huge demand $D$---any improvement of our result would imply an analogous
improvement for the $k$-median problem. 
(In this respect, our result is a noteworthy exception to the prevalent state of affairs 
for various other 
generalizations of \ufl and $k$-median---e.g., the data placement problem~\cite{BaevRS08},
\{matroid-, red-blue-\}
median~\cite{KrishnaswamyKNSS11,HajiaghayiKK12,CharikarL12,ChakrabartyS12}, 
$k$-facility-location~\cite{DevanurGKPSV05,GuptaT08}---where the best approximation ratio
for the problem is worse by a noticeable factor (compared to \ufl or $k$-median);
\cite{GoertzN11} is another exception.)  
Furthermore, {\em our analysis is tight} (up to $o(1)$ factors) because by
suitably setting $D$ in the reduction of~\cite{FriggstadS11}, we can ensure that our
local-search algorithm for \mfl coincides with the local-search algorithm for $k$-median
in~\cite{AryaGKMMP01} which has a tight approximation ratio of 3.

We also consider a weighted generalization of the problem (Section~\ref{extn}), wherein
each facility $i$ has a weight $w_i$ indicating the cost incurred per-unit distance moved
and the cost for moving $i$ to $s_i$ is $w_ic(i,s_i)$. (This can be used to model, for
example, the setting where different facilities move at different speeds.) Our analysis is
versatile and extends to this weighted generalization to yield the same performance guarantee. 
For the further generalization of the problem, where the facility-movement costs may be
arbitrary and unrelated 
to the client-assignment costs (for which a 9-approximation can be obtained via
LP-rounding; see ``Related work''), we show that local search based on multiple swaps has
a bad approximation ratio (Section~\ref{locgap}).  

The analysis leading to the approximation ratio of 3 (as also the simpler analysis in
Section~\ref{5apx} yielding a 5-approximation) crucially exploits the fact that we may
swap multiple locations in a local-search move. It is natural to wonder then if one can 
prove any performance guarantees for the local-search algorithm where we may only swap in
and swap out a single location in a local move. (Naturally, the single-swap algorithm is
easier to implement and thus may be more practical). In Section~\ref{oneswap}, we analyze
this single-swap algorithm and prove that it also has a constant approximation ratio. 

%\vspace{-1ex}
\paragraph{Our techniques.}
The analysis of our local-search procedure requires various novel ideas. 
As is common in the analysis of local-search algorithms, we identify a set of 
test swaps and use local optimality to generate suitable inequalities from these
test swaps, which when combined yield the stated performance guarantee. 
One of the difficulties involved in adapting standard local-search ideas to \mfl
is the following artifact: in \mfl, 
the cost of ``opening'' a set $S$ of locations is the cost of the min-cost perfect
matching of $\F$ to $S$, which, unlike other facility-location problems, is a highly
non-additive function of $S$ (and as mentioned above, we need to allow for the matching
from $\F$ to $S$ to change in non-local ways).  
In most facility-location problems with opening costs for which local search is known to
work, we may always swap in a facility used by the global optimum (by possibly swapping
out another facility) and easily bound the resulting change in {\em facility cost}, and
the main consideration is to decide how to reassign clients following the
swap in a cost-effective way;  
in \mfl we do not have this flexibility and need to carefully choose how
to swap facilities so as to ensure that there is a good matching of the facilities to
{their new destinations after a swap {\em and} there is a frugal reassignment of clients.}

This leads us to consider long relocation paths to re-match facilities to their new
destinations after a swap, which are of the form $(\ldots,s_i,o_i,s_{i'},\ldots)$,
where $s_i$ and $o_i$ are the locations that facility $i$ is moved to in the local and
global optimum, $S$ and $O$, respectively, and $s_{i'}$ is the $S$-location closest to
$o_i$. 
By considering a swap move involving the start and end locations of such a path $Z$, we
can obtain a bound on the movement cost of all facilities 
$i\in Z$ where $s_i$ is the start of the path or $o_i$ serves a large number of clients.
To account for the remaining facilities, we break up $Z$ into suitable
intervals, each containing a constant number of unaccounted locations which then
participate in a multi-location swap. 
This {\em interval-swap} move does not at first appear to be useful since we can only 
bound the cost-change due to this move in terms of a significant multiple of 
(a portion of) the cost of the local optimum! 
One of the novelties of our analysis is to show how we can {\em amortize} the cost of
such expensive terms and make their contribution negligible by considering multiple
different ways of covering $Z$ with intervals and averaging the inequalities obtained for
these interval swaps. These ideas lead to the proof of an approximation ratio of 5 for  
the local-search algorithm (Section~\ref{5apx}).

The tighter analysis leading to the 3-approximation guarantee (Section~\ref{3apx})
features another noteworthy idea, namely that of using ``recursion'' (up to
bounded depth) to identify a suitable collection of test swaps. 
We consider the tree-like structure created by the paths used in the 5-approximation
analysis, and (loosely speaking) view this as a recursion tree, using which we spawn 
off interval-swap moves by exploring this tree to a constant depth.
To our knowledge, we do not know of any analysis of a local-search algorithm that employs 
the idea of recursion to generate the set of test local moves (used to generate the
inequalities that yield the desired performance guarantee). We believe that this technique
is a notable contribution to the analysis of local-search algorithms that 
is of independent interest and will find further application. 

%\vspace{-1ex}
\paragraph{Related work.}
As mentioned earlier, \mfl was introduced by Demaine et al.~\cite{DemaineHMSOZ09} in the
context of movement problems. Friggstad and Salavatipour~\cite{FriggstadS11} designed the
first approximation algorithm for \mfl. They gave an 8-approximation
algorithm based on LP rounding by building upon the LP-rounding algorithm of
Charikar et al.~\cite{CharikarGTS02} for the $k$-median problem; this algorithm
works only however for the unweighted case. They also observed that 
there is an approximation-preserving reduction from $k$-median to \mfl. 
We recently learned that Halper~\cite{Halper10} proposed the same local-search
algorithm that we analyze. His work focuses on experimental results and leaves open the
question of obtaining theoretical guarantees about the performance of local search. 

Chakrabarty and Swamy~\cite{ChakrabartyS12} observed that \mfl, even with arbitrary 
movement costs is a special case of the matroid median
problem~\cite{KrishnaswamyKNSS11}. 
Thus, the approximation algorithms devised for matroid median independently
by~\cite{CharikarL12} and~\cite{ChakrabartyS12} yield an 8-approximation
algorithm for \mfl with arbitrary movement costs. 

There is a wealth of literature on approximation algorithms for
(metric) uncapacitated and capacitated facility location (\ufl and \cfl), the $k$-median
problem, and their variants; see~\cite{Shmoys04} for a survey on \ufl. 
Whereas constant-factor approximation algorithms for \ufl and $k$-median can be obtained
via a variety of techniques such as
LP-rounding~\cite{ShmoysTA97,Li11,CharikarGTS02,CharikarL12}, primal-dual
methods~\cite{JainV01a,JainMMSV02}, local
search~\cite{KorupoluPR00,CharikarG99,AryaGKMMP01}, all known $O(1)$-approximation  
algorithms for \cfl (in its full generality)
are based on local search~\cite{KorupoluPR00,ZhangCY03,BansalGG12}. 
We now briefly survey the work on local-search algorithms for facility-location problems. 

Starting with the work of~\cite{KorupoluPR00}, local-search techniques have been utilized
to devise $O(1)$-approximation algorithms for various facility-location problems.
Korupolu, Plaxton, and Rajaraman~\cite{KorupoluPR00} devised $O(1)$-approximation
for \ufl, and \cfl with uniform capacities, and $k$-median (with a blow-up in $k$). 
Charikar and Guha~\cite{CharikarG99}, and Arya et al.~\cite{AryaGKMMP01} both obtained a 
$(1+\sqrt{2})$-approximation for \ufl. 
The first constant-factor approximation for \cfl was obtained by P\'{a}l, Tardos, and
Wexler~\cite{PalTW01}, and after some improvements, the current-best approximation
ratio now stands at $5+\e$~\cite{BansalGG12}.
For the special case of uniform capacities, the analysis
in~\cite{KorupoluPR00} was refined by~\cite{ChudakW05}, 
and Aggarwal et al.~\cite{AggarwalABGGGJ10} obtain the current-best 3-approximation.
Arya et al.~\cite{AryaGKMMP01} devised a $(3+\e)$-approximation algorithm for $k$-median,
which was also the first constant-factor approximation algorithm for this problem based on
local search. 
Gupta and Tangwongsan~\cite{GuptaT08} (among other results) simplified the analysis
in~\cite{AryaGKMMP01}. We build upon some of their ideas in our analysis.

Local-search algorithms with constant approximation ratios have also been devised for
various variants of the above three canonical problems. Mahdian and
P\'{a}l~\cite{MahdianP03}, and Svitkina and Tardos~\cite{SvitkinaT10} consider settings
where the opening cost of a facility is a function of the set of clients served by
it. In~\cite{MahdianP03}, this cost is a non-decreasing function of the number of clients,
and in~\cite{SvitkinaT10} this cost arises from a certain tree defined on the 
client set. 
Devanur et al.~\cite{DevanurGKPSV05} and~\cite{GuptaT08} 
consider $k$-facility location, which is similar to $k$-median except that facilities
also have opening costs. Hajiaghayi et al.~\cite{HajiaghayiKK12} consider a special case
of the matroid median problem that they call the red-blue median problem. Most recently,
\cite{GoertzN11} considered a problem that they call the $k$-median forest problem, which
generalizes $k$-median, and obtained a $(3+\e)$-approximation algorithm.

%\vspace{-1ex}
\section{The local-search algorithm}
As mentioned earlier, to compute a solution to \mfl, we only need to determine the set of
final locations of the facilities, since we can then efficiently compute the best movement
of facilities from their initial to final locations, and the client assignments.
This motivates the following local-search operation. 
Given a current set $S$ of $k=|\F|$ locations, we can move to any other set $S'$
of $k$ locations such that $|S\sm S'|=|S'\sm S|\leq\lz$, where $\lz$ is some fixed value.
We denote this move by $\swap(S\sm S',S'\sm S)$. The local-search algorithm starts with an
arbitrary set of $k$ final locations. At each iteration, we choose the local-search move
that yields the largest reduction in total cost and update our final-location set
accordingly; 
if no cost-improving move exists, then we terminate. (To obtain polynomial running time,
as is standard, we modify the above procedure so that we choose a local-search move only
if the cost-reduction is at least $\e(\text{current cost})$.)

%\vspace{-1ex}
\section{Analysis leading to a 5-approximation} \label{5apx} \label{sec:5appx}
We now analyze the above local-search algorithm and show that it is a
$\bigl(5+o(1)\bigr)$-approximation algorithm. For notational simplicity, we assume that 
the local-search algorithm terminates at a local optimum; the modification to ensure
polynomial running time degrades the approximation by at most a $(1+\e)$-factor (see also
Remark~\ref{polyremk}). 

\begin{theorem} \label{5apxthm}
Let $F^*$ and $C^*$ denote respectively the movement and assignment cost of an optimal
solution. The total cost of any local optimum using at most $\lz$ swaps is at most 
$\Bigl(3+O\bigl(\frac{1}{\lz^{1/3}}\bigr)\Bigr)F^*+\Bigl(5+O\bigl(\frac{1}{\lz^{1/3}}\bigr)\Bigr)C^*$.
\end{theorem}

Although this is not the tightest guarantee that we obtain, 
we present this analysis first since it introduces many of the ideas that we
build upon in Section~\ref{3apx} to 
prove a tight approximation guarantee of $\bigl(3+o(1)\bigr)$ for the local-search
algorithm. 
For notational simplicity, we assume that all $d_j$s are 1. All our analyses carry over
trivially to the case of non-unit (integer) demands since we can think of a client $j$
having $d_j$ demand as $d_j$ co-located unit-demand clients.

%\vspace{-1ex}
\paragraph{Notation and preliminaries.}
We use $S=\{s_1,\ldots,s_k\}$ to denote the local optimum, where facility $i$ is moved to
final location $s_i\in S$. We use $O=\{o_1,\ldots,o_k\}$ to denote the (globally) optimal
solution, where again facility $i$ is moved to $o_i$. 
Throughout, we use $s$ to index locations in $S$, and $o$ to index locations in $O$.
Recall that, for a node $v$, $\sg(v)$ is the location in $S$ nearest to $v$. 
Similarly, we define $\sg^*(v)$ to be the location in $O$ nearest to $v$.
For notational similarity with facility location problems, we denote $c(i,s_i)$ by $f_i$,
and $c(i,o_i)$ by $f^*_i$. 
(Thus, $f_i$ and $f^*_i$ are the movement costs of $i$ in $S$ and $O$ respectively.)
Also, we abbreviate $c\bigl(j,\sigma(j)\bigr)$ to $c_j$, and $c\bigl(j,\sg^*(j)\bigr)$ to
$c^*_j$. 
Thus, $c_j$ and $c^*_j$ are the assignment costs of $j$ in the local and global
optimum respectively. 
(So $\mmfl(S)=\sum_{i\in\F}f_i+\sum_{j\in\D}c_j$.)
Let $D(s)=\{j\in\D:\sigma(j)=s\}$ be the set of clients assigned to the location $s\in S$,
and $D^*(o)=\{j\in\D:\sg^*(j)=o\}$. For a set $A\sse S$, we define 
$D(A)=\bigcup_{s\in A}D(s)$; we define $D^*(A)$ for $A\sse O$ similarly.   
Define $\capt(s) = \{o \in O : \sigma(o)=s\}$. We say that $s$ 
{\em captures} all the locations in $\capt(s)$. The following basic lemma 
will be used repeatedly. 
 
\begin{lemma} \label{reasgn}
For any client $j$, we have 
$c\bigl(j,\sg(\sg^*(j))\bigr)-c\bigl(j,\sg(j)\bigr)\leq 2c^*_j$.
\end{lemma}

\begin{proof}
Let $s=\sg(j),\ o=\sg^*(j),\ s'=\sg(o)$. The lemma clearly holds if $s'=s$. Otherwise, 
$
c(j,s')-c(j,s) \leq c(j,o) + c(o,s')-c(j,s) 
\leq c^*_j+c(o,s)-c(j,s) 
\leq c^*_j+c(o,j) = 2 c^*_j
$
where the second inequality follows since $s'$ is the closest location to $o$ in
$S$. 
\end{proof}

To prove the approximation ratio, we will specify a set of local-search
moves for the local optimum, and use the fact that none of these moves improve the cost to
obtain some inequalities, which will together yield a bound on the cost of the local 
optimum. We describe these moves by using the following digraph. 
Consider the digraph 
$\hG=\bigl(\F\cup S\cup O,\{(s_i,i),(i,o_i),(o_i,\sg(o_i))\}_{i\in\F}\bigr)$.
We decompose $\hG$ into a collection of node-disjoint (simple) paths $\Pc$
and cycles $\C$ as follows. 
Repeatedly, while there is a cycle $C$ in our current digraph, we add $C$ to $\C$, 
remove all the nodes of $C$ and recurse on the remaining digraph.
After this step, a node $v$ in the remaining digraph, which is acyclic, has: exactly one 
outgoing arc if $v\in S$; exactly one incoming and one outgoing arc if $v\in\F$; and exactly
one incoming, and at most one outgoing arc if $v\in O$. 
Now we repeatedly choose a node $v\in S$ with no incoming arcs, include the maximal path
$P$ starting at $v$ in $\Pc$, remove all nodes of $P$ and recurse on the remaining
digraph. Thus, each triple $(s_i,i,o_i)$ is on a unique path or cycle in $\Pc\cup\C$. 
Define $\cent(s)$ to be $o\in O$ such that $(o,s)$ is an arc in $\Pc\cup\C$; if $s$ has no
incoming arc in $\Pc\cup\C$, then let $\cent(s)=\nil$. 

We will use $\Pc$ and $\C$ to define our swaps. 
For a path $P=(s_{i_1},i_1,o_{i_1},\ldots,s_{i_r},i_r,o_{i_r})\in\Pc$, define 
$\spath(P)$ to be $s_{i_1}$ and $\epath(P)$ to be $o_{i_r}$. Notice that
$\sg(o_{i_r})\notin P$.
For each $s\in S$, let $\PP_c(s)=\{P:\epath(P) \in \capt(s)\}$, 
$T(s)=\{\spath(P): P \in \PP_c(s)\}$, and 
$H(s)=\{\epath(P):P\in \PP_c(s)\} = \capt(s)\setminus \cent(s)$. Note that $| \PP_c(s) | =
| T(s) |= |H(s)| = | \capt(s)| -1 $ for any $s\in S$ with $| \capt(s)| \geq 1$.
For a set $A\sse S$, define 
$T(A)=\bigcup_{s\in A} T(s),\ H(A)=\bigcup_{s\in A} H(s),\ \PP_c(A)=\bigcup_{s\in A}\PP_c(s)$.
  
A basic building block in our analysis, involves a {\em shift} along an
$s\leadsto o=o_{i'}$ sub-path $Z$ of some path or cycle in $\Pc\cup\C$. 
This means that we swap out $s$ and swap in $o$. 
We bound the cost of the matching between $\F$ and 
$S\cup\{o\}\sm\{s\}$ by moving each initial location $i\in Z,\ i\neq i'$ to    
$\sg(o_i)\in Z$ and moving $i'$ to $o_{i'}$. 
Thus, we obtain the following simple bound on the increase in movement cost due to this
operation: 
\begin{equation}
\shift(s,o) 
= \sum_{i\in Z}(f^*_{i}-f_{i})+\sum_{i\in Z: o_i\neq o} c\bigl(o_i,\sg(o_i)\bigr) 
\leq 2\sum_{i\in Z} f^*_{i}-c(o,\sg(o)).  
\label{shiftm}
\end{equation}
The last inequality uses the fact that 
$c\bigl(o_i,\sigma(o_i)\bigr)\leq c(o_i,s_i) \leq f^*_{i}+ f_{i}$ for all $i$. 
For a path $P\in\Pc$, we use $\shift(P)$ as a shorthand for 
$\shift\bigl(\spath(P),\epath(P)\bigr)$.

%\vspace{-1ex}
\subsection{The swaps used, and their analysis}
We now describe the local moves used in the analysis. 
We define a set of swaps such that 
each $o\in O$ is swapped in to an extent of at least one, and at most two.
We classify each location in $S$ as one of three types. 
Define $t=\bigl\lfloor{{\lz}^{1/3}}\bigr\rfloor$. We assume that $t\geq 2$. 

\begin{list}{$\bullet$}{\topsep=0.5ex \itemsep=0ex}
\item $S_0$: locations $s \in S$ with $|\capt(s)| = 0$.
\item $S_1$: locations $s \in S\sm S_0$ with $|D^*(\cent(s))|\leq t$ or $|\capt(s)|>t$. 
\item $S_2$: locations $s \in S$ with $|D^*(\cent(s))|>t$ and $0<|\capt(s)|\leq t$.
\end{list}

\noindent
Also define $S_3:=S_0\cup\{s \in S_1: |\capt(s)|\leq t\}$
(so $s\in S_3$ iff $|\capt(s)|\leq t$ and $|D^*(\cent(s))|\leq t\}$).  

To gain some intuition, notice that it is easy to generate a suitable inequality for a 
location $s\in S_0$: we can ``delete'' $s$ (i.e., if $s=s_i$, then do $\swap(s,i)$)
and reassign each $j\in D(s)$ to $\sg(\sg^*(j))$ (i.e., the location in $S$ closest
to the location serving $j$ in $O$). The cost increase due to this reassignment is at most 
$\sum_{j\in D(s)}2c^*_j$, and so this yields the inequality $f_i\leq\sum_{j\in D(s)}2c^*_j$. 
(We do not actually do this since we take care of the $S_0$-locations along with the
$S_1$-locations.)  
We can also generate a suitable inequality for a location $s\in S_2$ (see
Lemma~\ref{s2lem}) since we can swap in $\capt(s)$ and swap out $\{s\}\cup T(s)$. The cost
increase by this move can be bounded by $\sum_{P\in\PP_c(s)}\shift(P)$ and
$c\bigl(s,\cent(s)\bigr)$, and the latter quantity can be charged to
$\frac{1}{t}\sum_{j\in D^*(\cent(s))}(c_j+c^*_j)$; our definition of $S_2$ is tailored
precisely so as to enable this latter charging argument. 
Generating inequalities for the $S_1$-locations is more involved, and requires another
building block that we call an interval swap (this will also take care of the
$S_0$-locations), which we define after proving Lemma~\ref{s2lem}. 
We start out by proving a simple bound that one can obtain using a cycle in $\C$. 

\begin{lemma} \label{cycleswap}
For any cycle $Z\in\C$, we have 
$0\leq\sum_{i\in Z}\bigl(-f_i+f^*_i+c(o_i,\sg(o_i))\bigr)$.
\end{lemma}

\begin{proof}
Consider the following matching of $\F\cap Z$ to $S\cap Z$: we match $i$ to
$\sg(o_i)$. The cost of the resulting new matching is   
$\sum_{i\notin Z}f_i+\sum_{i\in Z}c(i,\sg(o_i))$ which should at least $\sum_i f_i$ since
the latter is the min-cost way of matching $\F$ to $S$. So we obtain
$0\leq\sum_{i\in Z}\bigl(-f_i+c(i,\sg(o_i))\bigr)
\leq\sum_{i\in Z}\bigl(-f_i+f^*_i+c(o_i,\sg(o_i))\bigr)$.
\end{proof}

\begin{lemma} \label{s2lem}
Let $s\in S_2$ and $o=\cent(s)$, and consider 
$\swap(X:=\{s\}\cup T(s),Y:=\capt(s))$. We have 
\begin{equation}
0 \leq \mmfl\bigl((S\setminus X) \cup Y\bigr) - \mmfl(S) 
\leq \sum_{\substack{P\in\PP_c(s) \\ i\in P}} 2f^*_i 
+\sum_{j\in D^*(o)}\Bigl(\tfrac{t+1}{t}\cdot c^*_j-\tfrac{t-1}{t}\cdot c_j\Bigr)
+\sum_{\substack{j\in D(\{s\}\cup T(s)) \\ j\notin D^*(o)}} 2c^*_j.
\label{s2ineq}
\end{equation}
\end{lemma}

\begin{proof}
We can view this multi-location swap as doing $\swap(\spath(P),\epath(P))$ for each
$P\in\PP_c(s)$ and $\swap(s,o)$ simultaneously.
(Notice that no path $P\in\PP_c(s)$ contains $s$, 
since $s=\sg\bigl(\epath(P)\bigr)\notin P$.) 
For each $\swap(\spath(P),\epath(P))$ the movement-cost increase is bounded by
$\shift(P)\leq\sum_{i\in P}2f^*_i$. For $\swap(s,o)$ we move the facility $i$, where
$s=s_i$, to $o$, so the increase in movement cost is at most 
$c(s,o)=c(\sg(o),o)\leq c(\sg(j),o)\leq c_j+c^*_j$ for every $j\in D^*(o)$. 
So since $|D^*(o)|>t$, we have $c(s,o)\leq\sum_{j\in D^*(o)}\frac{c_j+c^*_j}{t}$.
Thus, the increase in total movement cost is at most

We upper bound the change in assignment cost by reassigning the clients in 
$D^*(o)\cup D(X)$ as follows. 
We reassign each $j \in D^*(o)$ to $o$. Each $j\in D(X)\setminus D^*(o)$
is assigned to $\sg^*(j)$, if $\sg^*(j)\in Y$, and otherwise to
$s'=\sigma(\sigma^*(j))$. Note that $s'\notin X$: $s'\neq s$ since 
$\sg^*(j)\notin\capt(s)$, and $s'\notin T(s)$ since $\bigcup_{s''\in T(s)}\capt(s'')=\es$. 
The change in assignment cost for each such client $j$ is at most $2c^*_j$ by
Lemma~\ref{reasgn}. Thus the change in total assignment cost is at most
$\sum_{j\in D^*(o)} (c^*_j-c_j)+\sum_{j\in D(X)\setminus D^*(o)} 2c^*_j$. Combining this
with the bound on the movement-cost change proves the lemma.
\end{proof}

We now define a key ingredient of our analysis, called an {\em interval-swap} operation,
that is used to bound the movement cost of the $S_1$- and $S_0$-locations and the
assignment cost of the clients they serve. (We build upon this in Section~\ref{3apx} to
give a tighter analysis proving a 3-approximation.)
Let $S'=\{s'_1,\ldots,s'_r\}\subseteq S_0\cup S_1,\ r\leq t^2$ be a subset of at most
$t^2$ locations on a path or cycle $Z$ in $\Pc\cup \C$, where $s'_{q+1}$ is the next
location in $(S_0\cup S_1)\cap Z$ after $s'_q$. 
Let $O'=\{o'_1,\ldots,o'_r\}\sse O$
where $o'_{q-1}=\cent(s'_q)$ for $q=2,\ldots,r$ and $o'_r$ is an arbitrary location that
appears after $s'_r$ (and before $s'_1$) on the corresponding path or cycle.
Consider each $s'_q$. If $|\capt(s'_q)|>t$, choose a {\em random} path
$P\in\PP_c(s'_q)$ with probability $\frac{1}{|\PP_c(s'_q)|}$, and set
$X_q=\{\spath(P)\}$ and $Y_q=\{o'_q\}$. 
If $|\capt(s'_q)|\leq t$, set $X_q=\{s'_q\}\cup T(s'_q)$, and $Y_q=\{o'_q\}\cup H(s'_q)$.
Set $X=\bigcup_{q=1}^r X_q$ and $Y=\bigcup_{q=1}^r Y_q$. Note that $|X|=|Y|\leq t^3$ since 
$|X_q|=|Y_q|\leq t$ for every $q=1,\ldots,r$.
Notice that $X$ is a random set, but $Y=O'\cup H(S'\cap S_3)$ is deterministic.
To avoid cumbersome notation, we use $\swap(X,Y)$
to refer to the distribution of swap-moves that results by the random choices above,
and call this the {\em interval swap corresponding to $S'$ and $O'$}.
We bound the expected change in cost due to this move below. 
Let $\bon(s)$ be the indicator function that is 1 if $s\in S_3$ and 0 otherwise.

\begin{lemma} \label{s1lem}
Let $S'=\{s'_1,\cdots,s'_r\}\subseteq S_0\cup S_1,\ r\leq t^2$ and $O'$ be as given above. 
Let $o'_0:=\cent(s'_1)=o_{\hi}$, where $o'_0=\nil$ and $D^*(o'_0)=\es$ if $s'_1\in S_0$.
Consider the interval swap $\swap\bigl(X=\bigcup_{q=1}^r X_q,Y=\bigcup_{q=1}^r Y_q\bigr)$ 
corresponding to $S'$ and $O'$, as defined above.
We have 
\vspace{-1ex}
\begin{equation}
\hspace*{-4ex}
\begin{split}
0 & \leq\ \E{\mmfl\bigl((S\setminus X) \cup Y\bigr) - \mmfl(S)}\ \leq\
\sum_{q=1}^{r} \shift(s'_q,o'_q)+\sum_{{P\in\PP_c(S'), i\in P}}2f^*_i 
+\sum_{j\in D^*(O')} (c^*_j-c_j) \hspace*{-4ex} \\
& + \sum_{j\in D\left(T(S'\cap S_3)\cup (S'\cap S_3)\right)} 2c^*_j
+\sum_{j\in D\left(T(S'\sm S_3)\right)} \tfrac{2c^*_j}{t}
+\bon(s'_1)\sum_{j\in D^*(o'_0)} (f^*_{\hi}+f_{\hi}+c^*_j).
\end{split}
\label{t3}
\end{equation}
\end{lemma}

\begin{proof} 
Let $Z$ be the path in $\Pc$ or cycle in $\C$ such that $S'\cup O'\sse Z$.

We first bound the increase in movement cost.
The interval swap can be viewed as a collection of simultaneous 
$\swap(X_q,Y_q),\ q=1,\ldots,r$ moves.
If $X_q=\{\spath(P)\}$ for a random path $P\in \PP_c(s'_q)$, the movement-cost increase can be
broken into two parts. 
We do a shift along $P$, but move the last initial location on $P$ to $s'_q$, and then do
shift on $Z$ from $s'_q$ to $o'_q$. So the expected movement-cost change is
at most
$$
\frac{1}{|\PP_c(s'_q)|}\sum_{P\in\PP_c(s'_q)}\bigl(\shift(P)+c(\epath(P),s'_q)\bigr)+\shift(s'_q,o'_q)
\leq\frac{1}{|\PP_c(s'_q)|}\sum_{P\in\PP_c(s'_q),i\in P}2f^*_i+\shift(s'_q,o'_q)
$$
which is at most $\sum_{P\in\PP_c(s'_q),i\in P}2f^*_i+\shift(s'_q,o'_q)$.
Similarly, if $|\capt(s'_q)|\leq t$, we can break the movement-cost increase into
$\shift(P)\leq\sum_{i\in P}2f^*_i$ for all $P\in\PP_c(s'_q)$ and $\shift(s'_q,o'_q)$.
Thus, the total increase in movement cost is at most
\begin{equation}
\sum_{q=1}^{r} \shift(s'_q,o'_q)+\sum_{P\in\PP_c(S'),i\in P}2f^*_i.
\label{f3}
\end{equation}

Next, we bound the change in assignment cost by reassigning clients in 
$\hD=D^*(O')\cup D(X)$ as follows. 
We assign each client $j\in D^*(O')$ to $\sg^*(j)$. 
If $|\capt(s'_1)|>t$, then $s'_1\notin X$. 
For every client $j\in\hD\sm(D^*(O'))$, observe that either $\sg^*(j)\in Y$
or $\sg(\sg^*(j))\notin X$. To see this, let $o=\sg^*(j)$ and $s=\sg(o)$.  
If $o\notin Y$ then $s\notin S'\cap S_3$; also $s\notin T(S')$, and so $s\notin X$.  
So we assign $j$ to $\sg^*(j)$ if $\sg^*(j)\in Y$ and to $\sg(\sg^*(j))$ otherwise; the
change in assignment cost of $j$ is at most $2c^*_j$ (Lemma~\ref{reasgn}).

Now suppose $|\capt(s'_1)|\leq t$, so $s'_1\in X$. For each 
$j\in\hD\sm (D^*(O')\cup D^*(o'_o))$, we again have $\sg^*(j)\in Y$ or
$\sg(\sg^*(j))\notin X$, and we assign $j$ to $\sg^*(j)$ if $\sg^*(j)\in Y$ and to
$\sg(\sg^*(j))$ otherwise.  
We assign every $j\in\hD\cap D^*(o'_0)$ to $s_{\hi}$ (recall
that $o'_0=o_{\hi}$), and overestimate the resulting change in assignment cost by
$\sum_{j\in D^*(o'_o)}(c^*_j+f^*_{\hi}+f_{\hi})$. 
Finally, note that we reassign a client $j\in D(T(S'\sm S_3))\sm D^*(O')$ with probability
at most $\frac{1}{t}$ (since $\sg(j)\in X$ with probability at most $\frac{1}{t}$). 
So taking into account all cases, we can bound the change in total assignment cost
by  
\begin{equation}
\sum_{j\in D^*(O')} (c^*_j-c_j)
+\sum_{j\in D(T(S'\cap S_3)\cup(S'\cap S_3))} 2c^*_j
+\sum_{j\in D(T(S'\sm S_3))} \tfrac{2c^*_j}{t}
+\bon(s'_1)\sum_{j\in D^*(o'_0)} (f^*_{\hi}+f_{\hi}+c^*_j).
\label{cl3}
\end{equation}
In \eqref{cl3}, we are double-counting clients in 
$D\bigl(T(S')\cup(S'\cup S_3)\bigr)\cap D^*(O')$.
We are also overestimating the change in assignment cost of a
client $j\in D(X)\cap D^*(o'_0)$ 
since we include both the $\bon(s'_1)(c^*_j+f^*_{\hi}+f_{\hi})$ term, and the
$2c^*_j$ or $\frac{2c^*_j}{t}$ terms. 
Adding \eqref{f3} and \eqref{cl3}
yields the lemma.
\end{proof}

Notice that Lemma~\ref{s2lem} immediately translates to a bound on the assignment cost
of the clients in $D^*(\cent(s))$ for $s\in S_2$. 
In contrast, it is quite unclear how Lemma~\ref{s1lem} may be useful, since the
expression $\sum_{j\in D^*(o'_0)}(f^*_{\hi}+f_{\hi})$ in the RHS of \eqref{t3} may be as
large as $t(f^*_{\hi}+f_{\hi})$ (but no more since $|D^*(o'_0)|\leq t$ if $\bon(s'_1)=1$)
and it is unclear how to cancel the contribution of $f_{\hi}$ on the RHS. 
One of the novelties of our analysis is that
we show how to {\em amortize} such expensive terms and make their contribution negligible
by considering multiple interval swaps. We cover each path or cycle $Z$ in $t^2$ 
different ways using intervals comprising consecutive locations from $S_0\cup S_1$.
We then argue that averaging, over these $t^2$ covering ways, the inequalities
obtained from the corresponding interval swaps yields (among other things) a good bound on 
the movement-cost of the $(S_0\cup S_1)$-locations on $Z$ and the assignment cost of the
clients they serve.  

\begin{lemma} \label{corbigI}
Let $Z\in\Pc\cup\C$, $S'=\{s'_1,\ldots,s'_{r}\}=S_1\cap Z$, where $s'_{q+1}$ is the next
$S_1$-location on $Z$ after $s'_q$, and $O'=\{\cent(s'_1),\ldots,\cent(s'_r)\}$. 
Let $o'_r=\epath(Z)$ if $Z\in\Pc$ and $\cent(s'_1)$ otherwise.
For $r\geq t^2$, 
\begin{equation}
\begin{split}
0\ &\leq\ \sum_{i\in Z}\Bigl(\tfrac{t+1}{t}\cdot f^*_i-\tfrac{t-1}{t}f_i\Bigr)
+\sum_{P\in\PP_c(S'), i\in P} 2f^*_i 
+\sum_{j\in D^*(Z\cap O)}\Bigl(\tfrac{1}{t}\cdot c_j+\tfrac{t+1}{t^2}\cdot c^*_j\Bigr) \\
& + \sum_{j\in D^*\left(O'\cup\{o'_r\}\right)}(c^*_j-c_j)
+\sum_{j\in D\left(T(Z\cap S_3)\cup (Z\cap S_3)\right)}2c^*_j 
+\sum_{j\in D\left(T(S'\sm S_3)\right)}\tfrac{2c^*_j}{t}. 
\end{split}
\label{bndIlarge}
\end{equation}
\end{lemma}

\begin{proof}
We first define formally an interval of (at most) $t^2$ consecutive $(S_0\cup S_1)$
locations along $Z$. As before, let $o'_{q-1} =\cent(s'_q)$ for $q=1,\ldots,r$. 
For a path $Z$, define $s'_q=\spath(Z)$ for $q\leq 0$ and $s'_q=\nil$ for $q>r$.
Also define $o'_q=o'_0$ for $q\leq 0$ and $o'_q=\epath(Z)$ for $q\geq r$. 
If $Z$ is a cycle, we let our indices wrap around and be $\bmod ~r$, i.e.,
$s'_q=s'_{q\bmod r},\ o'_q=o'_{q\bmod r}$ for all $q$ (so $o'_r=o'_0=\cent(s'_1)$). 

For $1-t^2\leq h \leq r$, define $S'_h=\{s'_h,s'_{h+1},\ldots,s'_{h+t^2-1}\}$ to be an 
interval of length at most $t^2$ on $Z$. Define $O'_h=\{o'_h,o'_{h+1},\ldots,o'_{h+t^2-1}\}$.
Note that we have $1\leq |S'_h|=|O'_h| \leq t^2$ if $Z$ is a path, and 
$|S'_h|=|O'_h|=t^2$ if $Z$ is a cycle. 
Consider the collection of intervals, $\{S'_{-t^2+1},S'_{-t^2+2},\cdots,S'_r\}$. 
For each $S'_h, O'_h$, where $-t^2+1\leq h \leq r$, we consider the interval swap
$(X_h,Y_h)$ corresponding to $S'_h, O'_h$. We add the inequalities
$\frac{1}{t^2}\times$\eqref{t3} for all such $h$.
Since each $s'\in S'\cup\{s'_0\}$ participates in exactly $t^2$ such inequalities, 
and each $s'_h\in S'$ is the start of only the interval $S'_h$, we obtain the following.  
\begin{equation}
\begin{split}
0\ & \leq\ \sum_{q=0}^{r} \frac{1}{t^2}\cdot t^2\cdot \shift(s'_q,o'_{q})
+\sum_{P\in\PP_c(S'), i\in P} \frac{1}{t^2}\cdot t^2\cdot 2f^*_i \\ 
& + \sum_{j\in D^*(O'\cup\{o'_r\}))} \frac{1}{t^2}\cdot t^2\cdot (c^*_j-c_j)
+\sum_{j\in D\left(T(Z\cap S_3)\cup(Z\cap S_3)\right)}\frac{1}{t^2}\cdot t^2\cdot 2c^*_j
+\sum_{j\in D\left(T(S'\sm S_3)\right)}\frac{1}{t^2}\cdot t^2\cdot\frac{2c^*_j}{t} \\
& + \sum_{i:\sg(o_i)\in Z}\bon(\sg(o_i))\cdot\frac{1}{t^2}\cdot\sum_{j\in D^*(o_i)}(f^*_i+f_i+c^*_j).
\end{split}
\label{intsumbnd}
\end{equation}

Notice that the $S$-locations other than $s'_q$ on the $s'_q\leadsto o'_q$ sub-paths of $Z$
lie in $S_2$, and for each $i$ such that $\sg(o_i)\in Z\cap S_2$, we have
$c(o_i,\sg(o_i))\leq\sum_{j\in D^*(o_i)}\frac{c_j+c^*_j}{t}$. 
Thus, using \eqref{shiftm}, we have
\begin{equation}
\sum_{q=0}^{r} \shift(s'_q,o'_{q})=\sum_{i\in Z}(f^*_i-f_i)+\sum_{i:\sg(o_i)\in Z\cap S_2}c(o_i,\sg(o_i))
\leq \sum_{i\in Z}(f^*_i-f_i)
+\sum_{j\in D^*(Z\cap O)}\frac{c_j+c^*_j}{t}.
\label{sqoqbnd}
\end{equation}
Since $\bon(\sg(o))=1$ means that $\sg(o)\in S_3$, and so $|D^*(o)|\leq t$, we have
\begin{equation}
\sum_{i:\sg(o_i)\in Z\cap S_3}\sum_{j\in D^*(o_i)}\frac{f^*_i+f_i+c^*_j}{t^2}
\leq \sum_{i\in Z}\Bigl(\frac{f^*_i+f_i}{t}+\frac{\sum_{j\in D^*(o_i)} c^*_j}{t^2}\Bigr)
\leq \sum_{i\in Z}\frac{f^*_i+f_i}{t}+\sum_{j\in D^*(Z\cap O)}\frac{c^*_j}{t^2}.
\label{bonbnd}
\end{equation}

\noindent
Incorporating \eqref{sqoqbnd} and \eqref{bonbnd} in \eqref{intsumbnd}, and simplifying
yields the desired inequality. 
\end{proof}

For a path or cycle $Z$ where $|S_1\cap Z|<t^2$, we obtain an inequality similar to
\eqref{bndIlarge}. Since we can now cover $Z$ with a single interval, we never have
a client $j$ such that none of $\sg(j),\ \sg^*(j),\ \sg(\sg^*(j))$ are in our new set of
final locations. So the resulting inequality does not have any
$\frac{f^*_i+f_i}{t}+\frac{c^*_j}{t^2}$ terms.

\begin{lemma} \label{corsmallI}
Let $Z\in\Pc\cup\C$, $S'=\{s'_1,\ldots,s'_{r}\}=S_1\cap Z$, where $s'_{q+1}$ is the next
$S_1$-location on $Z$ after $s'_q$, and $O'=\{\cent(s'_1),\ldots,\cent(s'_r)\}$. 
Let $o'_r=\epath(Z)$ if $Z\in\Pc$ and $\cent(s'_1)$ otherwise.
For $r'<t^2$, 
\begin{equation}
\begin{split}
0\ & \leq\ \sum_{i\in Z}(f^*_i-f_i) + \sum_{P\in\PP_c(S'), i\in P} 2f^*_i 
+ \sum_{j\in D^*(Z\cap O)}\frac{c_j+c^*_j}{t} \\ 
& + \sum_{j\in D^*\left(O'\cup\{o'_r\}\right)}(c^*_j-c_j)
+ \sum_{j\in D\left(T(Z\cap S_3)\cup (Z\cap S_3)\right)} 2c^*_j 
+ \sum_{j\in D\left(T(S'\sm S_3)\right)}\tfrac{2c^*_j}{t}. 
\end{split}
\label{bndIsmall}
\end{equation}
\end{lemma}

\begin{proof} 
The proof is similar to that of Lemma~\ref{corbigI}, except that since we can cover
$Z$ with a single interval, we only need to consider a single (multi-location) swap.
We consider two cases for clarity.

\begin{list}{\arabic{enumi}.}{\usecounter{enumi} \addtolength{\leftmargin}{-3ex}
    \topsep=0.5ex \itemsep=0.5ex}
\item {\bf \boldmath $Z$ is a path, or $r>0$.}
As before, let $o'_{q-1} =\cent(s'_q)$ for $q=1,\ldots,r$. 
If $Z$ is a path, define $s'_0=\spath(Z)$. 
If $Z$ is a cycle, we again set $s'_q=s'_{q\bmod r},\ o'_q=o'_{q\bmod r}$ for all $q$.  
Consider the interval swap $(X,Y)$ corresponding to 
$S'\cup\{s'_0\}, O'\cup\{o'_r\}$. 
The inequality generated by this is
similar to \eqref{t3} except that we do not have any $(f^*_{\hi}+f_{\hi}+c^*_j)$ terms
since for client $j\in D(X)\cup D^*(Y)$, we always have that either $\sg^*(j)\in Y$ or
$\sg(\sg^*(j))\notin X$. Thus, \eqref{t3} translates to the following.
$$
0 \leq \sum_{q=0}^r\shift(s'_q,o'_q)+\sum_{P\in\PP_c(S'), i\in P}\negthickspace 2f^*_i
+\sum_{j\in D^*(O'\cup\{o'_r\})}\negthickspace (c^*_j-c_j)
+\sum_{j\in D\left(T(Z\cap S_3)\cup(Z\cap S_3)\right)}\negthickspace 2c^*_j
+\sum_{j\in D\left(T(S'\sm S_3)\right)}\negthickspace\tfrac{2c^*_j}{t}.
$$
Substituting 
$\sum_{q=0}^{r} \shift(s'_q,o'_{q})
\leq \sum_{i\in Z}(f^*_i-f_i)+\sum_{j\in D^*(Z\cap O)}\frac{c_j+c^*_j}{t}$ 
as in \eqref{sqoqbnd}
yields the stated inequality. 

\item {\bf\boldmath $Z$ is a cycle with $r=0$.}
Here, Lemma~\ref{cycleswap} yields 
$0 \leq \sum_{i\in Z}(f^*_i-f_i)+\sum_{j\in D^*(Z\cap O)}\frac{c_j+c^*_j}{t}$
(which is the special case of the earlier inequality with 
$s'_0=\nil=o'_r,\ S'=O'=Z\cap S_3=\es$).  \vspace{-4ex}
\end{list}
\end{proof}

\begin{proofof}{Theorem~\ref{5apxthm}}
We consider the following set of swaps.
\begin{list}{A\arabic{enumi}.}{\usecounter{enumi} \itemsep=0ex \topsep=0.5ex}
\item For every $s\in S_2$, the move $\swap\bigl(\{s\}\cup T(s),\capt(s)\bigr)$. 
\item For every path or cycle $Z$ with $|Z\cap S_1|\geq t^2$, the $\frac{1}{t^2}$-weighted
interval swaps as defined in Lemma~\ref{corbigI}.
\item For every path or cycle $Z$ with $|Z\cap S_1|<t^2$, the interval swap defined in
Lemma~\ref{corsmallI}. 
\end{list}

Notice that every location $o\in O$ is swapped in to an extent of at least 1 and at most
2. (By ``extent'' we mean the total weight of the inequalities involving $o$.)
To see this, suppose first $o=\epath(Z)$ for some path $Z$, then $o$ is involved to an
extent of 1 in the interval swaps for $Z$ in A2 or A3. In this case, we say that the
interval swap for $Z$ is {\em responsible} for $o$. Additionally, if 
$s=\sg(o)\in S_2$, then $o$ belongs to the multi-swap for $s$ in A1, else if $s\in S_3$
then $o$ is part of the interval swap for the path/cycle containing $s$ in A2 or A3. 
Now suppose $o=\cent(s)$. If $s\in S_2$, then $o$ is included in the multi-swap for 
$s$ in A1. We say that this multi-swap is responsible for $o$.
If $s\in S_1$, then $o$ is included in the interval swap for the path/cycle containing $s$
in A2 or A3, and we say that this interval swap is responsible for $o$.

Consider the compound inequality obtained by summing \eqref{s2ineq}, \eqref{bndIlarge}, and 
\eqref{bndIsmall} corresponding to the moves considered in A1, A2, and A3 respectively. 
The LHS of this inequality is 0. We now need to do some bookkeeping
to bound the coefficients of the $f^*_i, f_i, c^*_j, c_j$ terms on the RHS. We ignore
$o(1)$ coefficients like $\frac{1}{t},\ \frac{1}{t^2}$ in this bookkeeping, since for a
given $\{f^*_i,f_i,c^*_j,c_j\}$ term, such coefficients appear in only a constant number
of inequalities, so they have an $o(1)$ effect overall.     
Let $F$ and $C$ denote respectively the movement- and assignment- cost of the local
optimum. 

\smallskip
\noindent
{\bf Contribution from the $c^*_j$ and $c_j$ terms.}
First, observe that for each $o\in O$, we have labeled exactly one move involving $o$ as 
being responsible for it. Consider a client $j\in D(s)\cap D^*(o)$. 
Observe that $c^*_j$ or $c_j$ terms appear (with a $\Theta(1)$-coefficient) in an
inequality generated by a move if 
(i) $j$ is reassigned because the move is responsible for $o$; or
(ii) $s$ is swapped out (to an extent of 1) by the move
(so this excludes the case where $s\in T(s'),\ s'\in S_1\sm S_3$ and the move is
the interval swap for the path containing $s'$).   
If (i) applies, then the inequality generates the term $(c^*_j-c_j)$. 
If (ii) applies then the term $2c^*_j$ appears in the inequality. 
Finally, note that there are at most two inequalities for which (ii) applies:
\begin{list}{--}{\addtolength{\leftmargin}{-2ex} \topsep=0ex \itemsep=0ex}
\item If $s=\spath(Z)\in S_0$, then (ii) applies for the interval-swap move for $Z$. If
$s'=\sg(\epath(Z))\in S_2\cup S_3$, then (ii) again applies, for the multi-swap move for
$s'$ if $s'\in S_2$, or for the interval swap for the path containing $s'$ if 
$s'\in S_3$. 

\item If $s\in S_1\cap S_3$, then (ii) applies for the interval swap for the path
containing $s$.  

\item If $s\in S_2$, then (ii) applies for the multi-swap move for $s$.
\end{list}
So overall, we get a $5c^*_j-c_j$ contribution to the RHS, the bottleneck being the
two inequalities for which (ii) applies when 
$s\in\spath(Z)$ and $\sg(\epath(Z))\in S_2\cup S_3$.  

\smallskip
\noindent
{\bf Contribution from the $f^*_i$ and $f_i$ terms.}
For every $i\in\F$, the expression $(f^*_i-f_i)$ is counted once in the RHS of the
inequality \eqref{bndIlarge} or \eqref{bndIsmall} for the unique path or cycle $Z$
containing $i$. The total contribution of all these terms is therefore, $F^*-F$.   
The remaining contribution comes from expressions of the form 
$\sum_{P\in\PP_c(s),i\in P} 2f^*_i$ on the RHS of \eqref{s2ineq}, \eqref{bndIlarge}, and
\eqref{bndIsmall}. The paths $P$ involved in these expressions come from
$\PP_c(S_2)\cup\bigl(\bigcup_{Z\in\Pc\cup\C}\PP_c(Z\cap S_3)\bigr)\sse\Pc$.
Therefore, the total contribution of these terms is at most $2F^*$.

Thus, we obtain the compound inequality
\begin{equation}
0\leq \bigl(5+o(1)\bigr)C^*+\bigl(3+o(1)\bigr)F^*-\bigl(1-o(1)\bigr)(F+C)
\label{finineq}
\end{equation}
where the $o(1)$ terms are $O\bigl(\frac{1}{t}\bigr)=O\bigl(\lz^{-1/3}\bigr)$.
This shows that $F+C\leq\bigl(3+o(1)\bigr)F^*+\bigl(5+o(1)\bigr)C^*$.
\end{proofof}

\begin{remark} \label{polyremk}
A subtle point to note is that in the above analysis: 
(1) we consider only a polynomial number of swap moves (since there are at most $O(n^\lz)$
swap moves available at any point), and 
(2) we place a constant weight (of at most 1) on the inequality obtained from any given
swap move (when we take the weighted sum of the inequalities obtained from the various
swap moves).  
This is relevant because we can only ensure in polynomial time that we terminate at an
approximate local optimum. More precisely, for any polynomial $f(n)$ and an $\e>0$, we
obtain in polynomial time a solution with cost $F+C$ such that the change in cost due to
any local move is at least $-\frac{\e}{f(n)}\cdot(F+C)$ (instead of 0). But this means that 
the LHS of \eqref{finineq} is now $-\frac{\e}{f(n)}\cdot(F+C)N$, where $N$ is the total
weight placed on the inequalities generated from the various swap moves whose
suitable linear combination yields \eqref{finineq}. Therefore, since $N\leq f(n)$, this
only results in a $(1+\e)$-loss in approximation factor. 
\end{remark}

%\vspace{-1ex}
\section{Improved analysis leading to a 3-approximation} \label{3apx}

In this section, we improve the analysis from Section \ref{sec:5appx}. Specifically, we
prove the following theorem.

\begin{theorem} \label{thm:3appx}
The cost of a locally-optimal solution using $\lz$ swaps is at most 
$3+O\left(\sqrt\frac{\log\log\lz}{\log\lz}\right)$ times the optimum solution cost. 
\end{theorem}

To gain some intuition behind this tighter analysis, note that the {\em only} reason we
lost a factor of 5 in the previous analysis was because there could be locations 
$s=\spath(Z)\in S_0$ that are swapped out to an extent of 2; 
hence, there could be clients $j\in D(s)$ for which we ``pay'' $2c^*_j$ each time $s$
is swapped out, and also pay an additional $c^*_j-c_j$ term when $\sg^*(j)$ is swapped in.
To improve the analysis, we will consider a set of test swaps that swap out each 
location in $S$ to an extent of $1+o(1)$. 

The aforementioned bad case happens only when $s'=\sg(\epath(Z))\in S_2\cup S_3$, because
when we close (i.e., swap out) $s'$ as part of an interval swap 
or a multi-swap, we open (i.e., swap in) all the locations in $H(s')$, and we achieve
this via path swaps (i.e., $\shift$ moves) along paths in $\PP_c(s')$ that swap out
locations in $T(s')$ (for a second time). 
The main idea behind our refined analysis is to not perform such
path swaps, but instead to ``recursively'' start an interval swap on each path in
$\PP_c(s')$. 
Of course, we cannot carry out this recursion to arbitrary depth 
(since we can only swap a bounded number of locations), 
so we terminate the recursion at a depth of $t^2$. 
So, whereas an interval swap included at most $t^2$ $S_1$-locations on the main path or
cycle $Z$, we now consider a ``subtree'' swap obtained by aggregating interval swaps on
the paths in $\bigcup \PP_c(Z\cap S_3)$. 
A subtree swap can be viewed as a bounded-depth recursion tree where each
leaf to root path encounters at most $t^2$ locations in $S_1$. Because we no longer
initiate path swaps for $S_3$-locations, a leaf location $s''\in S_3$ in this recursion 
tree will not have any locations in $\capt(s'')$ opened. 
But we will slightly redefine the $S_1, S_2, S_3$ sets to ensure that 
$|D^*(\capt(s''))| \leq t$, and use the same trick that we did with interval swaps in
Section~\ref{5apx}: 
we {\em average} over different sets of subtree swaps (like we did with interval swaps in
Section~\ref{5apx}) to ensure that $s''$ is a leaf
location with probability at most $\frac{1}{t^2}$. This ensures that we incur, to an extent
of at most $\frac{1}{t}$, the cost $f^*_i+f_i+c(o_i,s'')$, where $o_i=\cent(s'')$, for
moving $j$ with $\sg(\sg^*(j))=s''$ from $s''$ to $s_i$. 

%\vspace{-1.5ex}
\paragraph{Notation.}
Let $t$ be an integer such that $\lz \geq t^2t^{t^2}+1$.
We prove that the local-search algorithm has approximation ratio $3+O(t^{-1})$. 
We redefine $S_0, S_1, S_2$ and $S_3$ as follows.
\begin{list}{$\bullet$}{\topsep=0.5ex \itemsep=0ex}
\item $S_0=\{s \in S: |\capt(s)|= 0\}$.
\item $S_1=\{s\in S\sm S_0: |D^*(\capt(s))|\leq t\text{ or }|\capt(s)|>t\}$.
\item $S_2=\{s \in S: |D^*(\capt(s))|>t,\ |\capt(s)|\leq t\}$.
\item $S_3=S_0\cup\{s\in S_1: |\capt(s)| \leq t\}$.
\end{list}
Clearly, $S=S_0\cup S_1\cup S_2$. 
We also redefine $\cent(s)$ to be the location in $\capt(s)$ {\em closest to $s$}. 

\begin{claim} \label{s2dist}
Let $s\in S_2$ and $o=\cent(s)$. 
Then $c(s,o)\leq\frac{1}{t}\sum_{j\in D^*(\capt(s))}(c_j+c^*_j)$. 
\end{claim}

\begin{proof}
We have $c(s,o)\leq c(s,o')$ for any $o'\in\capt(s)$, and $c(s,o')\leq c_j+c^*_j$ for any
$j\in D^*(o')$. Therefore, $c(s,o)\leq c_j+c^*_j$ for any $j\in D^*(\capt(s))$, and the
claim follows since $|D^*(\capt(s))|>t$ as $s\in S_2$. 
\end{proof}

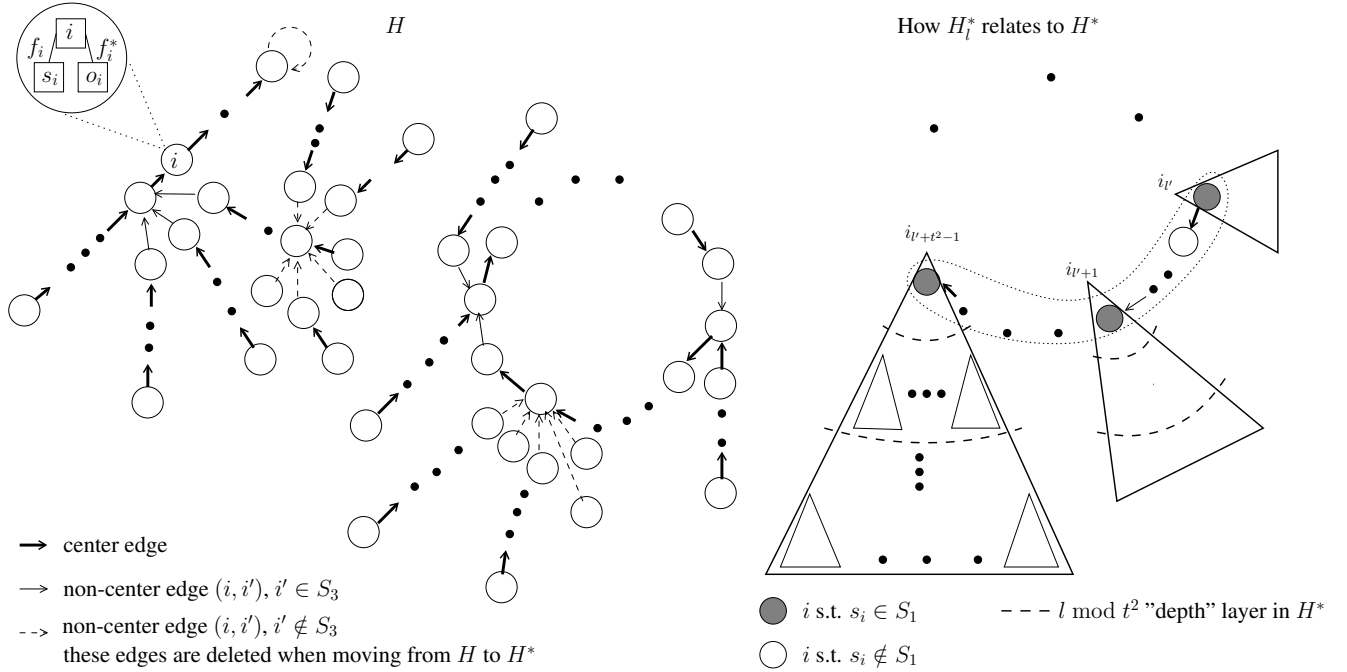
\begin{figure}[ht!]
\centerline{\resizebox{!}{3.5in}{\input{hgraph.pspdftex}}}
\label{fig:hgraph}
\end{figure}

It will be more convenient to work with 
the digraph $H=(\F,E)$ obtained from $\hG$ by contracting each triple
$\{s_i, i, o_i\}$ of nodes associated with a facility $i$ into a single node that we also
denote by $i$. Thus, 
$(i,i')$ is an arc in $E$ if $\sigma(o_i) = s_{i'}$ (it may be that $i = i'$). Note that $H$
may have self loops, and each node in $H$ has outdegree exactly 1 (counting self-loops) so
each component of $H$ looks like a tree with all edges oriented toward the root, except
the root is in fact a directed cycle (possibly a self-loop).  
The figure above illustrates this graph and some of the subgraphs and structures discussed
below. 

For brevity, we say that an edge $(i,i')$ in $H$ is a {\em center} edge if $o_i = \cent(s_{i'})$.
In the arc set $E'=\{(i,i')\in E: o_i=\cent(s_{i'})\})$, each node has indegree and
outdegree at most 1, so $E'$ consists of a collection node-disjoint paths $\Pc$ and cycles
$\C$. For a facility $i\in\Pc$, let $\Pc(i)$ denote the unique path in $\Pc$ containing
$i$. Let $\spath(i)$ and $\epath(i)$ denote the start and end facilities of $\Pc(i)$
respectively.  
For distinct facilities $i, i'$ with $s_i, s_{i'} \in S_0$, the paths $\Pc(i)$ and
$\Pc(i')$ are clearly vertex disjoint. 

Now define $H^*=\bigl(\F,E'\cup\{(i,i'): s_{i'}\in S_3, \sg(o_i)=s_{i'}\}\bigr)$; that
is, $H^*$ is the subgraph of $H$ with node-set $\F$ and edges $(i,i')$ of $E$ where
$(i,i')$ is a center edge or $s_{i'} \in S_3$.  
Call a node $i$ of $H^*$ a {\em root} if $i$ has no outgoing arc or $i$ lies on a directed
cycle in $H^*$. 

We consider an integer $1 \leq l \leq t^2$ and describe a set of swaps for each index
$l$. The inequalities for the swaps for different $l$ will be averaged in the final
analysis. We obtain $H^*_l$ by deleting the edges $(i,i')$ of $H^*$ where: 
$i$ is not on a cycle, $s_{i'} \in S_1$, and 
the number of facilities $i''$ with $s_{i''} \in S_1$ on the path between $i'$ and the
root of its component in $H^*$ (both included) is $l\bmod t^2$.  
We define a {\em subtree} of $H^*_l$ to be
an acyclic component of $H^*_l$, or a component that results by deleting the edges
of the cycle contained in a component of $H^*_l$. 

For a facility $i$, define 
$\cand(i)=\{i': o_{i'}\in\capt(s_i)\sm\{\cent(s_i)\},\ \text{$\not\exists i\leadsto i'$
  path in $H^*$}\}$.  
Note that $|\cand(i)|\geq|\capt(s_i)|-2$. The reason we define $\cand(i)$ is that we will
sometimes perform a shift along some path $Z\in\PP_c(s_i)$ to reassign the facilities on
$Z$ but we will not want this to interfere with the operations in the subtree of $H^*_l$ 
containing $i$.  
For a facility $i$ with $s_i \not\in S_2$, let $\next(i)$ be the facility obtained as
follows. Follow the unique walk from $i$ in $H$ using only center edges until the
walk reaches a node $i'$ with either no outgoing center edge, or the unique $(i',i'')$
center edge satisfies $s_{i''} \in S_1$; 
we set $\next(i) = i'$.

\begin{claim} \label{subtreesize}
The number of facilities $i$ with $s_i \in S_0 \cup S_1$ in any subtree of $H^*_l$ is at
most $t^{t^2}$.
\end{claim}

\begin{proof}
The facilities $i$ in such a subtree that are in $S_2$ have indegree and outdegree at
most 1.  Shortcutting past these facilities yields a tree with depth at most $t^2$ and
branching factor at most $t$.
\end{proof}

%\vspace{-2.5ex}
\paragraph{The test swaps.}
For a subtree $T$ of $H^*_l$, we describe a set of nodes $X_T$ to be swapped out and a
set of nodes $Y_T$ to be swapped in with $|X_T| = |Y_T| \leq t^{t^2}$. We do not actually
perform these swaps yet to generate the inequalities since we will have to combine some of
these swaps for various components. 

For each $i \in T$ with $s_i\in S_0\cup S_1$, we add the following location in $S$ to
$X_T$: if $s_i \in S_3$ we add $s_i$ to $X_T$; otherwise (so $s_i \in S_1 \setminus S_3$),
we choose any single $i'\in \cand(i)$ uniformly at random and add $s_{\spath(i')}$ to $X_T$. 
We also add $o_{\next(i)}$ to $Y_T$. 

When we say perform $\swap( X_T, Y_T )$, we specifically mean the following
reassignment of facilities.  For $s_i \in X_T$ with $s_i \in S_3$, we perform
$\shift(s_i, o_{\next(i)})$. For $s_i \in X_T$ with $s_i \in S_1\setminus S_3$, say $i'$ is
the facility in $\cand(i)$ for which $s_{\spath(i')}$ was added to $X_T$. Then we perform
$\shift(s_{\spath(i')}, o_{i'})$, move facility $i'$ from $o_{i'}$ to $s_i$, and
finally perform $\shift(s_i, o_{\next(i)})$.  
As always, each client is then assigned to its nearest final location. 
Lemma~\ref{lem:disjoint} implies that these shift operations charge different portions of
the local and global optimum.  

\begin{lemma} \label{lem:disjoint}
For a subtree $T$, all of the shift operations described for $\swap(X_T, Y_T)$ have
their associated paths being vertex disjoint.
\end{lemma}

\begin{proof} 
For any subtree $T$, the paths between $s_i$ and $o_{\next(i)}$ for the facilities 
$i\in T$ with $s_i \in S_0 \cup S_1$ are vertex-disjoint by definition of $\next(i)$.
Also, for any two distinct $i_1, i_2 \in T$, and any $i\in\cand(i_1),\  i'\in\cand(i_2)$, 
we have $\spath(i)\neq\spath(i')$, and so their associated paths $\Pc(i)$ and $\Pc(i')$
are also vertex-disjoint. 

Finally, consider any $i\in T$ with $s_i\in S_0\cup S_1$, and $i''\in T$ ($i''$ could be
$i$) with $s_{i''}\in S_1\sm S_3$. Consider the paths involved in
$\swap(s_i,o_{\next(i)})$ and $\swap(s_{\spath(i')},o_{i'})$, where $i' \in \cand(i'')$. 
Note that both of these paths consist of only center edges. Therefore, since each facility
has at most one incoming and one outgoing center edge, and $i'=\epath(i')$, if these paths are
not vertex-disjoint, then it must be that the path involved in $\swap(s_i,o_{\next(i)})$ 
is a subpath of the path involved in $\swap(s_{\spath(i')}, o_{i'})$. 
This means that $i$ and $i'$, and hence, $i, i', i''$, are all in the same component of
$H^*$. Also, the edge $(i',i'')$ is not in $H^*$ so $i'$ is the root of its component in
$H^*$. But then there is a path from $i''$ to $i'$, which contradicts that
$i'\in\cand(i'')$.   
\end{proof}

We need to coordinate the swaps for the various subtrees of $H^*_l$.  
Consider a component $Z$ in $H^*$. Let $C=\es$ if $Z$ is rooted at a node, otherwise 
let $C$ be its cycle of root nodes. 
Let $i_1,\ldots,i_k$ be the facilities on $C$
with $s_i\in S_1$, indexed by order of appearance on $C$ starting from an arbitrary
facility $i_1$ on $C$ ($k=0$ if $C=\es$). We consider four kinds of swaps.
\begin{list}{{\bf Type \arabic{enumi}.}}{\usecounter{enumi} \topsep=0.5ex \itemsep=0ex
    \leftmargin=0ex \settowidth{\labelwidth}{{\bf Type 4. }}
    \itemindent=\labelwidth \addtolength{\itemindent}{0.5ex}} 
\item 
If $1\leq k\leq t^2$, simultaneously do $\swap(X_T, Y_T)$ for all
subtrees $T$ rooted at some $i \in C$ with $s_i \in S_1$. 

\item Otherwise, if $k>t^2$, 
define $I_{l'} = \{i_{l'}, i_{l'+1}, \ldots, i_{l'+t^2-1}\}$ for all $l'=1,\ldots,k$ 
(where the indices are $\,\bmod k$).
Simultaneously perform $\swap( X_T, Y_T)$ for each
subtree $T$ rooted at a facility in $I_{l'}$. 
Reasoning similarly as in Lemma \ref{lem:disjoint} and noting that
the subtrees involved in a single type-1 or type-2 swap are all disjoint, we can see that   
all shift paths involved in a single type-1 or type-2 swap are vertex-disjoint.

\item
For each $i$ with $s_i \in S_2$, simultaneously perform $\swap(X_T, Y_T)$ for all
subtrees $T$ rooted at some $i'$ with $o_{i'} \in \capt(s_i) \setminus \{\cent(s_i)\}$.
At the same time, we also swap out $s_i$ and swap in $o_{i''}=\cent(s_i)$ for a
total of at most $t^{t^2+1} + 1 \leq p$ swaps. 
It may be that some (at most one) shift path in this swap includes $s_i$, but then we just move
$i''$ to $o_{i''}$ instead of $s_i$, and then move $i$ according to the shift operation.

\item Finally, for every other subtree $T$ of $H^*_l$ that was not swapped in the previous
cases, perform $\swap(X_T, Y_T)$ on its own. 
\end{list}

%\vspace{-1.5ex}
\paragraph{Analysis.}
We first bound the net client-assignment cost increase for any single one of these test
swaps. So, fix one such swap, let $\{T_r\}_{r=1}^k, k \leq t^2$ be the set of subtrees
involved in the swap, and let $B$ denote the set of facilities $i$
such that $o_i = \cent(\sigma(o_i))$ and $\sigma(o_i)$ is closed during the swap while $o_i$ is not opened.
So $B$ consists of facilities with a center edge to some leaf of some subtree $T_r$ or,
if the swap is of type 2, to the start of an interval $I_{l'}$.
For this swap, let $C_1 = \{j\in\D: \sigma^*(j) {\rm ~is~opened}\}$,
$C_2 = D^*(\{o_i: i\in B\})$, and 
$C_3 = \{j: \sigma(j) = s_i\in S_0\text{ and }\epath(i)\in\bigcup_r\bigcup_{i'\in T_r}\cand(i')\}$. 

\begin{lemma} \label{lem:clientbound}
The expected change in client-assignment cost for a test swap is
at most $\sum_{j \in C_1} (c^*_j - c_j) +  \sum_{j \in C_2} 2c^*_j +
\frac{1}{t-1}\sum_{j \in C_3} 2c^*_j +  2t\sum_{i \in B} \left(f^*_i + f_i\right)$.
Here, the expectation is over the random choices involved in selecting facilities from the
appropriate $\cand(.)$ sets.
\end{lemma}

\begin{proof}
After the swap, we move every $j \in C_1$ from $\sigma(j)$ to $\sigma^*(j)$
for a cost change of $c^*_j - c_j$. Every client in $j \in C_2\cup C_3$ for which
$\sigma(j)$ is closed is moved initially to $\sigma(\sigma^*(j))$ for a cost increase of
at most $2c^*_j$. 

Suppose $i$ is such that $\sigma(o_i) = \sigma(\sigma^*(j))$ and $o_i=\cent(\sg(o_i))$.
It may be that $\sigma(o_i)$ is still not open which means that $i \in B$. 
Note that either $s_i$ or $o_i$ is opened after the shift and we move every client that
was moved to $\sigma(o_i)$ to $s_i$ or $o_i$ (whichever is open).
This extra distance moved is at most $f^*_i + f_i + c(o_i, \sigma(o_i)) \leq 2f^*_i + 2f_i$.
Note that $i\in B$ implies that $\sg(o_i)\in S_3$, otherwise $\sg(o_i)$ would not have  
been closed down in the swap. So $|D^*(\capt(\sg(o_i)))| \leq t$, by definition of $S_3$,
and at most $t$ clients will be moved to either $s_i$ or $o_i$ in this manner. 

Finally, we note that while $j \in C_3$ may have $\sigma(j)$ being closed,
this only happens with probability at most $\frac{1}{t-1}$. 
\end{proof}

Now, we consider the following weightings of the swaps. First, for a particular index
$1 \leq l \leq t^2$ we perform all type 1, 3, and 4 swaps. For a
component of $H^*_l$ containing a cycle $C$, we perform all type-2 swaps
for the various intervals $I_{l'}$ for $C$ and weight the client and facility cost
change by $\frac{1}{t^2}$. Finally, these weighted bounds on the client and facility cost change
are averaged over all $1 \leq l \leq t^2$.

\begin{lemma} \label{lem:client}
The expected change in client-assignment cost under the weighting described above, is at most
$\sum_j 3c^*_j - c_j + O\left( \frac{1}{t} \right) \bigl(\sum_i (f^*_i + f_i) + \sum_j c^*_j\bigr)$.
\end{lemma}

\begin{proof}
For a fixed $l$, every client $j$ is in $C_1$ as in Lemma \ref{lem:clientbound} to an extent of 1;
either once in a type 1, 3, or 4 swap or exactly $t^2$ times among the type-2 swaps, each
of which is counted with weight $\frac{1}{t^2}$.
Similarly, every client $j$ is in $C_2$ to an extent of at most 1 and is in $C_3$ to an extent of
at most 1 over all swaps for this fixed $l$. Finally, we note each facility $i$ on a cycle in $H^*$
lies in the set $B$ for at most one offset $1 \leq l' \leq k$ for that cycle, so its contribution
$2t (f^*_i + f_i)$ to the bound is only counted with weight $\frac{1}{t^2}$ for this fixed $l$.

Lastly, every facility $i$ not on a cycle
in $H^*$ lies in $B$ for at most one index $l, 1 \leq l \leq t^2$ and, then, in only one swap
for that particular $l$. Since we average the bound over all indices $l$ between 1 and $t^2$,
the contribution $2t(f^*_i + f_i)$ of such $i$ is counted with weight only $\frac{1}{t^2}$.
\end{proof}

Next we bound the expected facility movement cost change.
Let $F'$ be the set of facilities $i$ that do not lie on a cycle in $H^*$
consisting solely of facilities $i'$ with $s_{i'} \in S_2$. 

\begin{lemma} \label{lem:fac}
{The expected change in movement cost (under the weighting described above)
is at most $\sum_{i \in F'} \left(f^*_i  - f_i\right) + 
O\left(\frac{1}{t}\right) \bigl(\sum_i f^*_i + \sum_j(c^*_j + c_j)\bigr)$.
}
\end{lemma}

\begin{proof} 
We consider two cases for a facility $i$. First, suppose $s_i \in S_0 \cup S_1$.
Then when $s_i$ is swapped out in a subtree during the shift from $s_i$ to $o_{\next(i)}$,
$i$ is moved to either $o_i$, if $i=\next(i)$, or to $\sigma(o_i)$, if $i\neq\next(i)$.
The latter case implies that $\sigma(o_i) \in S_2$. The total movement change is at most
$f^*_i - f_i$ if $i$ is moved to $o_i$ and is at most $f^*_i - f_i + c(o_i, \sigma(o_i))$
if $i$ is moved to $\sigma(o_i)$. 
Since $\sg(o_i)\in S_2$ and $o_i=\cent(\sg(o_i))$, by Claim~\ref{s2dist} we have that  
$c(o_i, \sigma(o_i)) \leq \frac{1}{t} \sum_{j \in D^*(\capt(\sigma(o_i)))} (c^*_j + c_j)$.

The only other time $i$ is moved is when 
$\epath(i)$ is randomly chosen from $\cand(i')$ for
some facility $i'$. But this happens with probability at most $\frac{1}{t-1}$.  In this
case, $i$ is shifted from $s_i$ to $\sigma(o_i)$. 
We do not necessarily have $\sigma(o_i)\in S_2$ in this case, but we can use the bound
$c(o_i, \sigma(o_i)) \leq f^*_i + f_i$ to 
bound the expected movement-cost change for $i$ in this case to be at most
$\frac{2f^*_i}{t-1}$. Overall, the expected movement-cost change for $i$ is at most 
$$
\Bigl(1+\frac{2}{t-1}\Bigr)f^*_i - f_i + 
\frac{1}{t-1} \sum_{j \in D^*\bigl(\capt(\sigma(o_i))\bigr)} (c^*_j+c_j).
$$

Next, we consider the case $s_i \in S_2$. Let $\cent(s_i)=o_{i'}$. When the swap
consisting of $i$ and all subtrees rooted at $\capt(s_i)\setminus\{o_{i'}\}$ is
performed, $i$ is moved from $s_i$ to $o_{i'}$ unless $i$ lies on a shift path during that
swap, in which case it is moved like in the shift. 
Since $s_i \in S_2$, we have $c(s_i,o_{i'})\leq\frac{1}{t}\sum_{j\in D^*(\capt(s_i))}(c^*_j+c_j)$.  
Unless $i$ lies on a cycle with no $S_1$-locations, that is, $i\notin F'$, $i$ lies
between $i''$ and $\next(i'')$ for exactly one $i''$, and 
$\shift(s_{i''},o_{\next(i'')})$ is performed to an extent of 1; this holds even if $s_i$
lies on a shift path during the corresponding type-3 swap involving $i$. 
All other times $i$ when is moved, it is due to the same reasons as in the previous case,
so the total change in movement cost for facility $i$ is at most
$$
\Bigl(1+\frac{2}{t-1}\Bigr)f^*_i - f_i 
+ \frac{1}{t}\sum_{j \in D^*(\capt(s_i))} (c^*_j +c_j) 
+ \frac{1}{t} \sum_{j \in  D^*\bigl(\capt(\sigma(o_i))\bigr)} (c^*_j + c_j).
$$

Adding up the appropriate expression for each facility accounts for the expected change in
total movement cost. 
\end{proof}

\begin{proofof}{Theorem \ref{thm:3appx}}
By local optimality, the change in total cost for every test swap (counting every random
choice) is nonnegative. By averaging over the various swaps, the expected change in
total cost is nonnegative, so the sum of the expressions in Lemmas~\ref{lem:client}
and~\ref{lem:fac} is nonnegative.
To generate an inequality involving a $-f_i$ term for facilities $i\notin F'$, 
we sum the bound given by Lemma~\ref{cycleswap} here 
over all cycles of $H^*$ involving only facilities $i$ with $s_i\in S_2$.  
This yields $0 \leq \sum_{i \not\in F'}\bigl(-f_i+f^*_i+c(o_i,\sigma(o_i))\bigr)$,
and we can bound $c(o_i, \sigma(o_i))$ by $\frac{1}{t}\sum_{j \in D^*(\capt(\sigma(o_i)))} (c^*_j + c_j)$.
Adding this to the inequality that the expected change in total cost is nonnegative gives 
$\left(1-O\left(\frac{1}{t}\right)\right) (C+F) \leq \left(3 +
O\left(\frac{1}{t}\right)\right)C^* + \left(1 + O\left(\frac{1}{t}\right)\right) F^*$.
\end{proofof}

%\vspace{-2ex}
\section{Extension to the weighted case} \label{extn}
The analysis in Section~\ref{3apx} (as also the proof of the 5 approximation) extends
easily to the weighted generalization, wherein each facility $i$ has a weight $w_i\geq 0$
and the cost of moving $i$ to $s$ is given by $w_ic(i,s)$, to yield the same
$\bigl(3+o(1)\bigr)$-approximation guarantee. With the exception of one small difference
in the analysis, this requires only minor changes in the arguments. We discuss these
briefly in this section.

Unless otherwise stated, the same
notation from Section \ref{3apx} is used in this section.  We emphasize that $f^*_i$ and
$f_i$ now refer to the weighted movement cost of facility $i$ in the global or local
optimum, respectively.  So, $f^*_i = w_i \cdot c(i, o_i)$ and $f_i = w_i\cdot c(i,s_i)$.  

One difference in notation is that the definition of $S_1$ is slightly revised to this
weighted setting: $s_i \in S_1$ if $|\capt(s_i)| > t$, or $0 < |\capt(s_i)| \leq t$ and
$|D^*(\capt(s_i))| \leq \max\{w_i, w_{i'}\} \cdot t$, where $i'$ is such that $o_{i'} =
\cent(s_i)$ (equivalently, $(i',i)$ is a center edge in $H$).  If all facility weights are
1, then this definition of $S_1$ agrees with the definition in Section \ref{3apx}.
Similarly, we say $s_i \in S_2$ if $|\capt(s_i)| \leq t$ and $|D^*(\capt(s_i))| >
\max\{w_i, w_{i'}\} \cdot t$. Under these definitions, similar to Claim~\ref{s2dist}, we 
now have that $w_i\cdot c(s_i,o_{i'})\leq\frac{1}{t}\sum_{j\in D^*(\capt(s_i))}(c_j+c^*_j)$ 
(since $c(o_i,\sigma(o_i)) \leq c^*_j + c_j$ for any $j \in D^*(\capt(s_i))$ as
before, and $|D^*(\capt(s_i))|> w_i t$).

We consider the same set of test swaps and the same averaging of the inequalities
generated by these swaps.  When a test swap is performed, we consider the same shift
and reassignment of facilities to generate the inequalities.  In most cases, we also move
the clients in the same way as before with the exception that if a client $j$ has all of
$\sigma(j), \sigma^*(j)$ and $\sigma(\sigma^*(j))$ being closed, then we do not
necessarily send it to $s_{i}$ where $i$ is such that 
$o_{i}=\sg^*(j)$. This is discussed in Lemma~\ref{lem:w_clientbound}.

As in the discussion before Lemma~\ref{lem:clientbound},
we consider a swap involving subtrees $\{T_r\}_{r = 1}^k$.
Let $B$ be as before, and let $B'$ be the set of facilities $i$ such that
$i$ is a leaf in some $T_r$ or, if the swap is a type-2 swap, that $i$ is the
first facility in $I_{l'}$. Note that $i\in B$ if and only if the unique $(i,i')$
arc in $H^*$ is a center arc with $i' \in B'$. We let $C_1, C_2$, and $C_3$
also be defined as in Section \ref{3apx}.

\begin{lemma} \label{lem:w_clientbound}
The expected change in client assignment cost for a test swap is
at most $\sum_{j \in C_1} (c^*_j - c_j) +  \sum_{j \in C_2} 2c^*_j +
\frac{1}{t-1}\sum_{j \in C_3} 2c^*_j +  4t\sum_{i \in B \cup B'} \left(f^*_i + f_i\right).$
\end{lemma}

\begin{proof}
Consider one particular swap.  As in the proof of Lemma~\ref{lem:clientbound}, we move 
$j\in C_1$ to $\sigma^*(j)$ and $j \in C_2 \cup C_3$ to $\sigma(\sigma^*(j))$ and bound
their assignment cost change in the same way. As before, it may
be that for some of these clients $j \in C_2 \cup C_3$ we have that $\sigma(\sigma^*(j))$
was closed in the swap. For such clients, we do the following slight variant
of the reassignment that was done in the proof of Lemma \ref{lem:clientbound}.

Suppose $(i,i')$ is the center edge such that $\sigma^*(j)=o_{i}$ for a
client $j \in C_2 \cup C_3$, and $s_{i'}$ is not open.
If $w_{i} \geq w_{i'}$, then we send $j$ to either $s_{i}$ or $o_{i}$.
As in the proof of Lemma~\ref{lem:clientbound}, one of these must be open
and the total cost of moving $j$ from $s_{i'}$ to either $s_{i}$ or $o_{i}$
is at most $2c(i,s_{i}) + 2c(i,o_{i})$. Otherwise, if $w_{i'} > w_{i}$
then we send $j$ to either $o_{i'}$ or $\sigma(o_{i'})$ (one of them must be open).
The distance from $s_{i'}$ to either $o_{i'}$ or $\sigma(o_{i'})$ is bounded by
$2c(i',s_{i'}) + 2c(i',o_{i'})$.

We conclude by noting that each facility $\hi \in B$ has at most $w_{\hi} \cdot t$ clients
sent to either $s_{\hi}$ or $o_{\hi}$ from $\sigma(o_{\hi})$ in the manner just described,
because $\sg(o_{\hi})$ must be in $S_3$. 
Similarly, each $\hi \in B'$ has at most  $w_{\hi} \cdot t$ clients sent to either $o_{\hi}$
or $\sigma(o_{\hi})$ from $s_{\hi}$ in the manner described above, since $s_{\hi}\in S_3$. 
So, the total client movement charged to $i \in B \cup B'$ this way is at most
$4tw_i(c(i,s_i) + c(i,o_i)) = 4tf^*_i + 4tf_i$.
\end{proof}

Using the same weighting of the swaps as in Section \ref{3apx} we get the following
bound on the contribution of the client movement cost changes over these swaps.
The proof is nearly identical, except we notice that a facility $i'$ lies in the $B'$-set
for various swaps to an extent of at most $\frac{1}{t^2}$ (under this weighting), since
the facility $i$ such that $(i,i')$ is a center edge lies in some $B$-set to an extent of
at most $\frac{1}{t^2}$. 

\begin{lemma} \label{lem:w_client}
The expected total client assignment cost change, weighted in the described manner, is at most
$\sum_j 3c^*_j - c_j + O\left( \frac{1}{t} \right) \bigl(\sum_i (f^*_i + f_i) + \sum_j c^*_j\bigr)$.
\end{lemma}

The contribution of the facility movement costs is bounded in essentially the same way as
in Lemma \ref{lem:fac}.  We just provide the details on how to account for the weights of
the facilities. As before, let $F'$ be the set of facilities $i$ that do not lie on a
cycle in $H^*$ consisting solely of facilities $i'$ with $s_{i'} \in S_2$. 

\begin{lemma} \label{lem:w_fac}
\mbox{
The expected change in movement cost 
is at most $\sum_{i \in F'} \left(f^*_i  - f_i\right) + 
O\left(\frac{1}{t}\right) \bigl(\sum_i f^*_i + \sum_j (c^*_j + c_j)\bigr)$.
}
\end{lemma}

\begin{proof}
When $\shift(s,o)$ is performed, we move facilities $i$ from $s_i$ to $\cent(o_i)$.  If
this shift was performed during a path swap, then the movement-cost change for a
facility $i$ moved in this shift is at most $w_i c(i,o_i) + w_i c(o_i, \sigma(o_i)) -
w_i c(i, s_i) \leq 2w_i c(i, o_i) = 2f^*_i$ so the same bound used before applies.

If such a shift was performed along a path in a subtree, then we did not want to bound
$c(o_i, \sigma(o_i))$ by $c(i, s_i) + c(i, o_i)$ because we do not want to cancel the
contribution of $-c(i, s_i)$ to the bound.  However, this only happened when $\sigma(o_i)
\in S_2$ so we can use the fact that $|D^*(\capt(\sigma(o_i)))|$ is large and that 
$c(o_i,\sigma(o_i)) \leq c^*_j + c_j$ for any $j \in D^*(\capt(\sigma(o_i)))$.  In our
setting, as noted earlier, the movement cost $w_i \cdot c(o_i, \sigma(o_i))$ can be
bounded by $\frac{1}{t}\sum_{j \in D^*(\capt(\sigma(o_i)))}(c^*_j + c_j)$, 
which is the same upper bound we used in the unweighted case.

Finally, the only other time we moved a facility was from some $s_i \in S_2$ to
$\cent(s_i)$. The cost of this move is now $w_i \cdot c(s_i, \cent(s_i))$ which can also
be bounded by $\frac{1}{t} \sum_{j \in D^*(\capt(s_i))}(c^*_j + c_j)$ using the same
argument as in the previous paragraph. 
So, all bounds for the unweighted facility
movement cost increase averaged over the swaps also hold in the weighted case.
\end{proof}

Finally, we remark that the same bound for the facility movement cost for facilities on a
cycle with only $S_2$ facilities holds for the weighted case, again using arguments like
in the proof of Lemma \ref{lem:w_fac} to bound $w_i \cdot c(o_i, \sigma(o_i))$. Thus, the
proof of Theorem \ref{thm:3appx} is adapted to prove the following result for the weighted
case.

\begin{theorem}
The cost of a locally-optimal solution using $\lz$ swaps is at most 
$3+O\left(\sqrt\frac{\log\log p}{\log p}\right)$ times the optimum solution cost in
weighted instances of mobile facility location.
\end{theorem}

\section{The single-swap case} \label{oneswap}
We show in this section that the local-search algorithm has a constant approximation
guarantee also when $\lz=1$ (which corresponds to $t=1$ in Section~\ref{5apx}), even in
the weighted setting. This requires a different analysis than in
Section~\ref{5apx} since we now no longer have the luxury of amortizing the ``expensive''
terms in an interval swap via multi-location swaps. 
The approximation factor we obtain in the analysis below is large, but we emphasize that 
we have not sought to optimize this constant. Also, we remark that the analysis can be
significantly simplified and improved in the unweighted setting.

Recall that $f^*_i=w_i\cdot c(i,o_i)$ and $f_i=w_i\cdot c(i,s_i)$. 
We use much of the same notation as in Section~\ref{5apx}. 
The digraph $\hG$, its decomposition into paths $\Pc$ and cycles $\C$, and the definition 
of $\cent(s)$ are as in Section~\ref{5apx}. Thus, for a path $P\in\Pc$, we have
$\sg(\epath(P))\notin P$. Define $S_0=\{s \in S: |\capt(s)| = 0\}$. 
Let $s_i\in S\sm S_0$ with $o_{i'}=\cent(s_i)$. We place $s_i$ in $S_1$ if
$|D^*(o_{i'})|\leq 1.5\max\{w_i,w_{i'}\}$ or $|\capt(s_i)|>1$; otherwise we place $s_i$ in
$S_2$. Also define $S_3:=S_0\cup \{s\in S_1: |\capt(s)|\leq 1\}$.
Let $n^*_o=|D^*(o)|$ for $o\in O$.

\begin{lemma} \label{1sws2lem}
Let $s_i\in S_2$ and $o=\cent(s_i)$, and consider $\swap(s_i,o)$. We have
\begin{equation}
0\leq\mmfl\bigl((S\sm\{s_i\})\cup\{o\}\bigr)-\mmfl(S)
\leq\sum_{j\in D^*(o)}\Bigl(\frac{5}{3}\cdot c^*_j-\frac{1}{3}\cdot c_j\Bigr)
+\sum_{j\in D(s_i)\sm D^*(o)}2c^*_j.
\label{1sws2}
\end{equation}
\end{lemma}

\begin{proof}
We reassign clients in $D^*(o)$ to $o$, and each client $j$ in $D(s_i)\sm D^*(o)$ to
$\sg(\sg^*(j))$ incurring a total assignment-cost change of 
$\sum_{j\in D^*(o)}(c^*_j-c_j)+\sum_{j\in D(s)\sm D^*(o)}2c^*_j$.
The change in movement cost is at most 
$$
w_i c(o,s_i)\leq w_i\cdot\frac{\sum_{j\in D^*(o)}(c_j+c^*_j)}{n^*_o}\leq
\sum_{j\in D^*(o)}\tfrac{2}{3}(c_j+c^*_j)
$$ 
where the last inequality follows since $s_i\in S_2$. Adding this to the expression for
the change in assignment cost yields the lemma.
\end{proof}

In Lemmas~\ref{pclem}--\ref{cyclem}, we generate inequalities that will allow us to bound
the total movement cost, and the total assignment cost of clients in $S_0\cup S_1$. We use
the following notation for this sequence of lemmas. Given a path or cycle $Z\in\Pc\cup\C$,
let $S'_Z=\{s'_1,\ldots,s'_r\}=S_1\cap Z$, where $s'_{q+1}$ is the next $S_1$-location on
$Z$ after $s'_q$. Let $o'_{q-1}=\cent(s'_q)$ for $q=1,\ldots,r$, and
$O'_Z=\{o'_0,\ldots,o'_{r-1}\}$. 
Let $s'_q=s_{i_q}$ and $o'_q=o_{\hi_q}$ for $q=0,\ldots,r$ (see Fig.~\ref{labeling}).  
Let $Z_q$ denote the $s'_q\leadsto o'_q$ subpath of $Z$.
Consider an $o_i\rightarrow s_{i'}\rightarrow i'\rightarrow o_{i'}$ subpath of $Z$. 
Let $A_Z$ consist of all such $s_{i'}$ where $s_{i'}\in S'\cap S_3$ and $n^*_{o_i}\leq 1.5w_i$.
Let $B_Z$ consist of all such $s_{i'}$ where $s_{i'}\in S'\cap S_3$ and $n^*_{o_i}>1.5w_i$.
Note that if $s_{i'}\in B_Z$ then $w_{i'}>w_i$ and $n^*_{o_i}\leq 1.5w_{i'}$.
Let $C_Z$ consist of all such $s_{i'}$ where $n^*_{o_i}\leq 1.5w_i$.
Clearly, $A_Z\sse C_Z\sse S'_Z$ and $C_Z\cap B_Z=\es$.
When $Z$ is clear from the context (as in Lemmas~\ref{pclem}--\ref{cyclem}), we drop
the subscript $Z$ from $S'_Z, O'_Z, A_Z, B_Z, C_Z$.

\begin{figure}[ht!]
\centerline{\resizebox{!}{1.5in}{\input{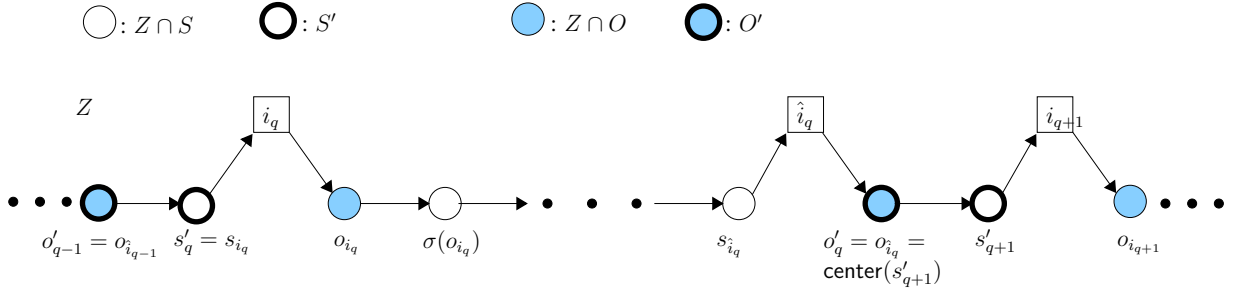}}}
\caption{The clear circles are locations in $S$; the shaded circles are locations in
$O$. The clear and shaded circles with thick borders depict locations in $S'$ and $O'$
respectively.}   
\label{labeling}
\end{figure}

\begin{lemma} \label{pclem}
Let $Z\in\Pc\cup\C$ be such that $r=|S'|=0$. Then,
\begin{align}
\text{\parbox[t]{0.88in}{if $Z\in\Pc$ with \\ $s'_0=\spath(Z)$ \\ $o'_r=\epath(Z)$}} 
&\ & 0\ & \leq\ \sum_{i\in Z}(f^*_i-f_i)
+\sum_{j\in D^*((Z\cap O)\sm\{o'_r\})}\negthickspace\negthickspace\negthickspace \tfrac{2}{3}(c_j+c^*_j) 
+\sum_{j\in D^*(o'_r)}\negthickspace (c^*_j-c_j)
+\sum_{j\in D^*(s'_0)}\negthickspace 2c^*_j. 
\label{pcineq1} \\[1ex]
\text{\parbox{0.88in}{if $Z\in\C$}}
&\ & 0\ & \leq\ \sum_{i\in Z}(f^*_i-f_i)
+\sum_{j\in D^*(Z\cap O)}\tfrac{2}{3}(c_j+c^*_j). 
\label{pcineq2}
\end{align}  
\end{lemma}

\begin{proof}
First, suppose $Z$ is a path.
For every $o_i\in Z,\ o_i\neq o'_r$, we have that $\sg(o_i)\in S_2$, and so we have 
$w_i\cdot c(o_i,\sg(o_i))\leq\sum_{j\in D^*(o_i)}\frac{2}{3}(c_j+c^*_j)$. Thus, the move
$\swap(s'_0,o'_r)$ yields the following inequality, which implies \eqref{pcineq1}.
$$
0\leq \sum_{i\in Z}(f^*_i-f_i)
+\sum_{\substack{i:\sg(o_i)\in Z\cap S_2 \\ j\in D^*(o_i)}}\tfrac{2}{3}(c_j+c^*_j)
+\sum_{j\in D^*(o'_r)}(c^*_j-c_j)+\sum_{j\in D^*(s'_0)\sm D^*(o'_r)}2c^*_j.
$$
For a cycle $Z$, analogous to Lemma~\ref{cycleswap}, we have 
$0\leq\sum_{i\in Z}\bigl(-f_i+f^*_i+w_i\cdot c(o_i,\sg(o_i))\bigr)$, and this coupled with
the above bound on $w_i\cdot c(o_i,\sg(o_i))$ for $i\in Z$ yields \eqref{pcineq2}. 
\end{proof} 

\begin{lemma} \label{pathlem1}
Let $Z\in\Pc$ (with $S', O', A, B, C$ as defined above). 
Let $s'_0=\spath(Z)$, and $O''=O'\cup\{o'_r=\epath(Z)\}$.
Suppose there is no $q$ such that $|Z_q\cap\F|=1,\ s'_q\in B,\ s'_{q+1}\in C$. Then,
\begin{equation}
\begin{split}
0\ \leq & \sum_{i\in\Pc_c(Z)} 5f^*_i+\sum_{i\in Z} (10.75f^*_i-f_i)
+\sum_{j\in D(s'_0)} 4.5c^*_j+\sum_{j\in D((S'\cap S_3)\cup H(Z\cap S))} 5c^*_j \\
& + \sum_{j\in D^*(Z\cap O'')} (6.5c^*_j-c_j)
+\sum_{j\in D^*((Z\cap O)\sm O'')}\tfrac{34}{6}(c_j+c^*_j). 
\end{split}
\label{pathineq1}
\end{equation}
\end{lemma}

\begin{proof}
We derive \eqref{pathineq1} by taking a weighted combination of inequalities
generated from various swap moves. We assume $r>0$ (otherwise \eqref{pcineq1} implies
\eqref{pathineq1}). 
For a predicate $R$, let $\bon(R)$ be the indicator function that is 1 if $R$ is true and  
0 otherwise.
To keep notation simple, we follow the convention that a non-existent object (like
$o'_{r+1}$) is $\nil$, and $\nil\notin T$ for every set $T$ (e.g., $o'_{r+1}\notin T$ for 
every $T$.) Also, $D(\nil)=D^*(\nil)=\es$, and $w_{\nil}=1$. 

\begin{list}{{\bf \arabic{enumi})}}{\usecounter{enumi} \settowidth{\labelwidth}{{\bf 5)}}
    \leftmargin=\labelwidth \addtolength{\leftmargin}{1ex}}
\item 
The first swap move is $\swap(s'_0,o'_r)$, but we bound the change in cost slightly
differently. We only reassign clients in $D(s'_0)$, obtaining the 
inequality 
\begin{equation}
0\leq\sum_{i\in Z}(f^*_i-f_i)+\sum_{i\in Z: \sg(o_i)\in Z}w_i\cdot c(o_i,\sg(o_i))+\sum_{j\in D(s'_0)}2c^*_j.
\label{pshift}
\end{equation}
Define $\al_i=1-\frac{2n^*_{o_i}}{3w_i}$. Note that $\al_i\in [0,1]$ if $\sg(o_i)\in A\cup C$.  
Consider each edge $(o_i,\sg(o_i))\in Z$. 
If $\sg(o_i)\notin A$ and $s_i\notin B$, we simply bound $w_i\cdot c(o_i,\sg(o_i))$ by
$f_i+f^*_i$.  
Otherwise, if $\sg(o_i)\in C$, we use the bound 
$w_i\cdot c(o_i,\sg(o_i))\leq\al_i(f_i+f^*_i)+\sum_{j\in D^*(o_i)}\tfrac{2}{3}(c_j+c^*_j)$
which is valid since $\sum_{j\in D^*(o_i)}\tfrac{2}{3}(c_j+c^*_j)\geq
\frac{2n^*_{o_i}}{3w_i}\cdot w_i\cdot c(o_i,\sg(o_i))$.
If $\sg(o_i)\notin C$ and $s_i\in B$, then $n^*_{o_i}>1.5w_i$ and we bound 
$w_i\cdot c(o_i,\sg(o_i))$ by $\sum_{j\in D^*(o_i)}\tfrac{2}{3}(c_j+c^*_j)$.
Incorporating these bounds in \eqref{pshift}, we obtain the following inequality.
\begin{equation}
\begin{split}
0\ \leq\ & (f^*_{\hi_r}-f_{\hi_r})
+\sum_{\substack{i\in Z: \sg(o_i)\in Z\sm A \\ s_i\notin B}} \negthickspace\negthickspace 2f^*_i
+\sum_{\substack{i\in Z: \sg(o_i)\in A\text{ or} \\ \sg(o_i)\in C\sm A, s_i\in B}}
\negthickspace\negthickspace\bigl(2f^*_i-(1-\al_i)f_i\bigr)
+\sum_{\substack{i\in Z: \sg(o_i)\in Z\sm C \\ s_i\in B}}\negthickspace\negthickspace(f^*_i-f_i) \\
& + \sum_{\substack{i\in Z: \sg(o_i)\in A\text{ or }s_i\in B \\ j\in D^*(o_i)}}\tfrac{2}{3}(c_j+c^*_j)
+\sum_{j\in D(s'_0)}2c^*_j.
\end{split}
\label{ineq1}
\end{equation}

\item 
For every $q=0,\ldots,r$ such that $s'_q\in A\cup S_0$, we consider the move 
$\swap(s'_q,o'_q)$. We move each facility $i\in Z_q$ to $\sg(o_i)$ if
$o_i\neq o'_q$ and to $o'_q$ otherwise. Note that for every facility $i$ such that  
$o_i, \sg(o_i)\in Z_q$, we have $\sg(o_i)\in S_2$, and so
$w_i\cdot c(o_i,\sg(o_i))\leq\sum_{j\in D^*(o_i)}\frac{2}{3}(c_j+c^*_j)$. 
We reassign all clients in $D^*(o'_q)$ to $o'_q$, and reassign each client
$j$ in $D(s'_q)\sm\bigl(D^*(o'_q)\cup D^*(o'_{q-1})\bigr)$ to $\sg(\sg^*(j))$. 
The resulting change in assignment cost is at most
$\sum_{j\in D^*(o'_q)}(c^*_j-c_j)+\sum_{j\in D(s'_q)}2c^*_j$. 
If $s'_q\in A$, we reassign all clients in $D^*(o'_{q-1})$ to
$s_{\hi_{q-1}}$ and bound the resulting change in assignment cost by  
$\sum_{j\in D^*(o'_{q-1})}c^*_j+\frac{n^*_{o'_{q-1}}}{w_{\hi_{q-1}}}(f_{\hi_{q-1}}+f^*_{\hi_{q-1}})$. 
Therefore, if $s'_q\in A\cup S_0$, we obtain the inequality
\begin{equation}
\begin{split}
0\ \leq\ & \sum_{i\in Z_q}(f^*_i-f_i)+\frac{n^*_{o'_{q-1}}}{w_{\hi_{q-1}}}(f_{\hi_{q-1}}+f^*_{\hi_{q-1}}) \\
& + \sum_{\substack{i\in Z_q: \sg(o_i)\in Z_q \\ j\in D^*(o_i)}}\tfrac{2}{3}(c_j+c^*_j)
+\sum_{j\in D^*(o'_q)}(c^*_j-c_j)+\sum_{j\in D^*(o'_{q-1})}c^*_j
+\sum_{j\in D(s'_q)}2c^*_j.
\end{split}
\label{Aineq}
\end{equation} 

\item 
For every $q$ such that $s'_q\in B$, we again consider $\swap(s'_q,o'_q)$. 
We move facilities on $Z_q$ and reassign clients in 
$D^*(o'_q)\cup\bigl(D(s'_q)\sm D^*(o'_{q-1})\bigr)$ as in case 2).
We reassign clients in $D^*(s'_q)\cap D^*(o'_{q-1})$ to 
$o_{i_q}$ if $o'_q=o_{i_q}$, and to $\sg(o_{i_q})$ otherwise.
The assignment-cost change due to this latter reassignment is at most
$\frac{n^*_{o'_{q-1}}}{w_{i_q}}\cdot (f_{i_q}+f^*_{i_q})$ if $o'_q=o_{i_q}$, and 
$\frac{n^*_{o'_{q-1}}}{w_{i_q}}\cdot (f_{i_q}+f^*_{i_q})+\sum_{j\in D^*(o_{i_q})}(c_j+c^*_j)$
otherwise, where to obtain the latter inequality we use the fact that
$n^*_{o'_{q-1}}\leq 1.5w_{i_q}\leq n^*_{o_{i_q}}$ since $s'_q=s_{i_q}\in B,\ \sg(o_{i_q})\in S_2$  
and $c(o_{i_q},\sg(o_{i_q}))\leq c_j+c^*_j$ for all $j\in D^*(o_{i_q})$. 
We obtain the inequality
\begin{equation}
\begin{split}
0\leq & \sum_{i\in Z_q}(f^*_i-f_i)+\frac{n^*_{o'_{q-1}}}{w_{i_q}}(f_{i_q}+f^*_{i_q}) \\
& + \sum_{\substack{i\in Z_q: \sg(o_i)\in Z_q \\ j\in D^*(o_i)}}
\negthickspace\negthickspace\tfrac{2}{3}(c_j+c^*_j)
+\sum_{j\in D^*(o'_q)}\negthickspace(c^*_j-c_j)
+\sum_{j\in D(s'_q)}2c^*_j
+\bon(o_{i_q}\notin O'')\negthickspace\negthickspace
\sum_{j\in D^*(o_{i_q})}\negthickspace(c_j+c^*_j).
\end{split}
\label{Bineq}
\end{equation}

\item 
For every $q$ such that $s'_q\in S'\sm S_3$, we pick some arbitrary $P\in\Pc_c(s'_q)$  
and consider $\swap(\spath(P),o'_q)$.
We analyze this by viewing this as a combination of: (i) a shift along $P$, (ii) moving
$i$ from $o_i$ to $s'_q$, where $o_i=\epath(P)$, and (iii) a shift along the
appropriate subpath of $Z$. The resulting inequality we obtain is therefore
closely related to \eqref{Aineq}, \eqref{Bineq}. We incur an additional
$\sum_{i\in P}2f^*_i$ term for the change in movement cost due to (i) and (ii), and
$\sum_{j\in D(\spath(P))}2c^*_j$ for reassigning clients in $D(\spath(P))$. 
Also, since $s'_q$ is no longer swapped out, we do not incur any terms that correspond to
reassigning clients in $D(s'_q)$. 
Thus, we obtain the following inequality:
\begin{equation}
0\leq \sum_{i\in P}2f^*_i+\sum_{i\in Z_q}(f^*_i-f_i)
+\sum_{\substack{i\in Z_q: \sg(o_i)\in Z_q \\ j\in D^*(o_i)}}\tfrac{2}{3}(c_j+c^*_j)
+\sum_{j\in D^*(o'_q)}(c^*_j-c_j)+\sum_{j\in D(\spath(P))}2c^*_j. 
\label{S1ineq} 
\end{equation}
\end{list}

We are finally ready to derive \eqref{pathineq1}. 
An inequality subscripted with an index, like \sub{Aineq}, denotes that inequality for
that particular index. It will be useful to define
\refstepcounter{equation}
\label{segineq}
$$
\sub{segineq}\ \equiv\ 
\bon(s'_q\in A\cup S_0)\cdot\sub{Aineq}+\bon(s'_q\in B)\cdot\sub{Bineq}
+\bon(s'_q\notin S_3)\cdot\sub{S1ineq}
$$
We take the following linear combination. 
\begin{equation}
\underbrace{2.25\times\eqref{ineq1}}_\text{\normalsize{part 1}}
\ +\ \underbrace{\sum_q 2.5\times\sub{segineq}}_\text{\normalsize{part 2}}.
\label{cmpineq1}
\end{equation}
The LHS of \eqref{cmpineq1} is 0. We analyze the contribution from the 
$f^*_i, f_i, c^*_j, c_j$ terms to the RHS. 

Facilities $i\notin Z$ contribute at most $\sum_{i\in\Pc_c(Z)} 5f^*_i$.
Consider $i\in Z$. If $o_i=\epath(Z)=o'_r$, we pick up $4.75(f^*_i-f_i)$ from parts 1 and 2,
and we may pick up an additional $2.5\cdot\frac{n^*_{o'_{r-1}}}{w_i}(f_i+f^*_i)$ from part 2 if
$s'_r\in B,\ i=i_r$. So overall, we obtain a contribution of at most $8.5f^*_i-f_i$. 
Next suppose $\sg(o_i)\in Z$. If $\sg(o_i)\in Z\sm A$ and $s_i\notin B$, we gather
$4.5f^*_i$ from part 1 and $2.5(f^*_i-f_i)$ from part 2, so the total contribution is
$7f^*_i-2.5f_i$. 
Suppose $\sg(o_i)\in A$. Notice then that $s_i\notin B$, otherwise if $s_i=s'_q$, then 
$Z_q\cap\F=\{i\}$, which contradicts our assumptions.
Recall that $\al_i=1-\frac{2n^*_{o_i}}{3w_i}$.
We gather $2.25\bigl(2f^*_i-(1-\al_i)f_i\bigr)$ from part 1, and 
$2.5\bigl[f^*_i-f_i+\frac{n^*_{o_i}}{w_i}(f_i+f^*_i)\bigr]\leq 
6.25f^*_i-2.5f_i\bigl(1-\frac{n^*_{o_i}}{w_i}\bigr)$ from part 2. 
Thus, we gather at most 
$10.75f^*_i-\bigl(2.5-\frac{n^*_{o_i}}{w_i}\bigr)f_i\leq 10.75f^*_i-f_i$.
If $\sg(o_i)\in Z\sm A$ and $s_i\in B$, then note that actually $\sg(o_i)\in Z\sm C$. 
We gather $2.25(f^*_i-f_i)$ from part 1, and 
$2.5\bigl[f^*_i-f_i+\frac{n^*_{\cent(s_i)}}{w_i}(f_i+f^*_i)\bigr]$ from part 2, so the total
contribution is at most $8.5f^*_i-f_i$.  

We now bound the $\{c^*_j, c_j\}$-contribution.
Part of this is 
$\sum_{j\in D(s'_0)} 4.5c^*_j+\sum_{j\in D((S'\cap S_3)\cup H(Z\cap S))} 5c^*_j$.
We proceed to analyze the remaining contribution.
The clients whose remaining contribution is non-zero are of two types:
(i) clients in $D^*(Z\cap O'')\cup D(Z\cap S_3)$, which are reassigned when a location in
$O''$ is swapped in or a location in $S''$ is swapped out; and
(ii) clients in $D^*(o_i)$, where $o_i\in Z\cap O$, which are charged when we
bound the cost $w_i\cdot c(o_i,\sg(o_i))$ 
or when $s_i\in B$ is swapped out.
Consider a client $j\in D^*(o'_q)$. Let $i=\hi_q$. Its (remaining) part-1 contribution is 
$1.5(c_j+c^*_j)$ if $\sg(o_i)\in A$ or $s_i\in B$, and 0 otherwise. 
The part-2 contribution is at most $2.5(c^*_j-c_j)+2.5c^*_j$ (this happens when 
$\sg(o_i)\in A$); so the total (remaining) contribution is at most $6.5c^*_j-c_j$.  
For $j\in D(s'_q)\sm D^*(Z\cap O'')$, where $s'_q\in Z\cap S_3$, we know that
$j\notin D^*(\capt(s'_q))$, and so we have already accounted for its contribution of at
most $5c^*_j$ above.
Finally, consider $j\in D^*(o_i)$, where $o_i\notin O''$. Note that $\sg(o_i)\in S_2$. We
gather $1.5(c_j+c^*_j)$ from part 1 if $s_i\in B$ and 0 otherwise, and 
$\frac{25}{6}(c_j+c^*_j)$ from part 2 if $s_i\in B$ and $\frac{5}{3}(c_j+c^*_j)$
otherwise; so in total we gather at most $\frac{34}{6}(c_j+c^*_j)$.

Putting everything together, \eqref{cmpineq1} leads to inequality \eqref{pathineq1}.
\end{proof}

\begin{lemma} \label{pathlem2}
Let $Z\in\Pc$, $s'_0=\spath(Z)$ and $o'_r=\epath(Z)$. We have
\begin{equation}
\begin{split}
0\ \leq & \sum_{i\in\Pc_c(Z)}32f^*_i+\sum_{i\in Z}(92.5f^*_i-f_i)
+\sum_{j\in D(s'_0)} 42c^*_j+\sum_{j\in D((S'\cap S_3)\cup H(Z\cap S))} 32c^*_j \\
& +\sum_{j\in D^*(Z\cap O'')}(41c^*_j-c_j)
+\sum_{j\in D^*((Z\cap O)\sm O'')}\frac{122}{3}(c_j+c^*_j).
\end{split}
\label{pathineq2}
\end{equation}
\end{lemma}

\begin{proof}
We focus on the case where there exists an index $q$ such that $|Z_q\cap\F|=1$,
$s'_q\in B$, and $s'_{q+1}\in A\cup C$ as otherwise, \eqref{pathineq1} 
immediately implies \eqref{pathineq2}. 
This case is significantly more involved, in part because when $s'_q\in B,\ s'_{q+1}\in A$
and $o'_q=o_{i_q}$, we accrue both 
the term $\frac{n^*_{\cent(s'_q)}}{w_{i_q}}\cdot (f_{i_q}+f^*_{i_q})$ in \eqref{Bineq}
when $s'_q$ is swapped out, and the term $\frac{n^*_{o'_q}}{w_{i_q}}\cdot (f_{i_q}+f^*_{i_q})$
in \eqref{Aineq} when $s'_{q+1}$ is swapped out. Hence, there is no way of combining
\eqref{ineq1}--\eqref{S1ineq} to get a compound inequality having both $-f_{i_q}$ and
$-\sum_{j\in D^*(o'_q)}c_j$ on the RHS. 
In order to achieve this, we define a structure called a {\em block}, comprising multiple 
$Z_q$ segments, using which we define additional moves that swap in $o'_q$ but
swap out neither $s'_q$ nor $s'_{q+1}$, so that the extent to which $o'_q$
is swapped in exceeds the extent to which $s'_q$ or $s'_{q+1}$ are swapped out. 

We call a set $\{s'_q,s'_{q+1},\ldots,s'_u\}$ of consecutive $S'$-locations a {\em block},
denoted by $\blk_{qu}$, if 
(i) $|Z_\ell\cap\F|=1$ and $s'_\ell\in B$ for all $\ell=q+1,\ldots,u$, 
(ii) $s'_{u+1}\in A\cup C$ (recall that if $s'_{u+1}$ is non-existent, then this condition
is not satisfied), and 
(iii) $|Z_q\cap\F|>1$ or $s'_q\notin B$. 
Note that by definition, any two blocks correspond to disjoint subpaths of $Z$. 
We say that $s, s'\in S\cap Z$ are {\em siblings}, denoted by $s\sib s'$, if $s, s'\in B$  
and they belong to a common block; note that this means that neither $s$ nor $s'$ is
at the start of a block. We use $s\nsib s'$ to denote that $s$ and $s'$ are not siblings.   
Let $O''=O'\cup\{o'_r\}$.

Before defining the additional swap moves for each block, we first reconsider
$\swap(s'_0,o'_r)$ and account for the change in cost due to this move differently to come
up with a slightly different inequality than \eqref{ineq1}. We again start with
\eqref{pshift}, and bound $w_i\cdot c(o_i,\sg(o_i))$ in different ways for an
$(o_i,\sg(o_i))$ edge of $Z$. We use 
$$
w_i\cdot c(o_i,\sg(o_i))\leq
\begin{cases}
%(i)
f_i+f^*_i; & 
\text{if $\sg(o_i)\notin A,\ s_i\notin B$, or $s_i\sib\sg(o_i)$} \\
%(ii) 
\al_i(f_i+f^*_i)+\sum_{j\in D^*(o_i)}\tfrac{2}{3}(c_j+c^*_j); & 
\text{if $\sg(o_i)\in A$ or $\sg(o_i)\in C\sm A,\ s_i\in B$} \\
%(iii)
\sum_{j\in D^*(o_i)}\tfrac{2}{3}(c_j+c^*_j); &
\text{if $sg(o_i)\notin C,\ s_i\in B$, and $s_i\nsib\sg(o_i)$}.
\end{cases}
$$
Incorporating this in \eqref{pshift} yields the following inequality.
\begin{equation}
\begin{split}
0\ \leq\ & (f^*_{\hi_r}-f_{\hi_r})
+\sum_{\substack{i\in Z: s_i\sib\sg(o_i) \text{ or} \\ \sg(o_i)\in Z\sm A, s_i\notin B}}  
\negthickspace\negthickspace 2f^*_i
+\sum_{\substack{i\in Z: \sg(o_i)\in A\text{ or} \\ \sg(o_i)\in C\sm A, s_i\in B}}
\negthickspace\negthickspace\bigl(2f^*_i-(1-\al_i)f_i\bigr) \\
& +\sum_{\substack{i\in Z: \sg(o_i)\in Z\sm C \\ s_i\in B, s_i\nsib\sg(o_i)}}
\negthickspace\negthickspace(f^*_i-f_i)
+ \sum_{\substack{i\in Z: \sg(o_i)\in A\text{ or} \\ s_i\in B, s_i\nsib\sg(o_i)}}
\sum_{j\in D^*(o_i)}\tfrac{2}{3}(c_j+c^*_j)+\sum_{j\in D(s'_0)}2c^*_j.
\end{split}
\label{ineq2}
\end{equation}

For every block $\blk_{qu}$, we consider the swap moves $\swap(s'_\ell,o'_{\ell+1})$ for
all $\ell=q+1,\ldots,u-1$. For each such $\ell$, we obtain the inequality 
\begin{equation}
0\leq\frac{n^*_{o'_{\ell-1}}}{w_{i_\ell}}\bigl(f_{i_\ell}+f^*_{i_\ell}\bigr)
+(f^*_{i_\ell}-f_{i_\ell})+(f^*_{i_{\ell+1}}-f_{i_{\ell+1}})
+\sum_{j\in D^*(o'_{\ell+1})}\negthickspace\negthickspace(c^*_j-c_j)
+\sum_{j\in D(s'_\ell)}\negthickspace\negthickspace2c^*_j
+\sum_{j\in D^*(o'_\ell)}\negthickspace\negthickspace\tfrac{5}{3}(c_j+c^*_j)
\label{BBineq}
\end{equation}
where we bound $w_{i_\ell}\cdot c(o'_\ell,s'_{\ell+1})$ by 
$\sum_{j\in D^*(o'_\ell)}\frac{2}{3}(c_j+c^*_j)$.
Also consider $\swap(s'_q,o'_{q+1})$ if $s'_q\in A\cup S_0\cup B$ or
$\swap(\spath(P),o'_{q+1})$ if $s'_q\in S'\sm S_3$, where $P$ is some arbitrary path in
$\Pc_c(s'_q)$. Note that if $s'_q\in B$, then $o_{i_q}\neq o'_q$. This yields the
inequality 
\begin{equation}
\begin{split}
0\ \leq\ & \bigl(f^*_{i_{q+1}}-f_{i_{q+1}}\bigr)+\sum_{i: o_i,\sg(o_i)\in Z_q\cup Z_{q+1}} 2f^*_i
+\sum_{j\in D^*(o'_{q+1})}(c^*_j-c_j) \\
& +\bon(s'_q\in A\cup S_0)
\Biggl(\sum_{j\in D(s'_q)}2c^*_j
+\frac{n^*_{o'_{q-1}}}{w_{\hi_{q-1}}}(f_{\hi_{q-1}}+f^*_{\hi_{q-1}})
+\sum_{j\in D^*(o'_{q-1})} c^*_j\Biggr) \\
& +\bon(s'_q\in B)
\Biggl(\sum_{j\in D(s'_q)}2c^*_j
+\frac{n^*_{o'_{q-1}}}{w_{i_q}}(f_{i_q}+f^*_{i_q})
+\sum_{j\in D^*(o_{i_q})}(c_j+c^*_j)\Biggr) \\
& +\bon(s'_q\in S'\sm S_3)
\Biggl(\sum_{i\in P}2f^*_i+\sum_{j\in D(\spath(P))} 2c^*_j\Biggr).
\end{split} 
\label{fBineq}
\end{equation}

We now derive \eqref{pathineq2} by taking the following linear combination of
\eqref{Aineq}--\eqref{fBineq}:  
\begin{equation}
\underbrace{21\times\eqref{ineq2}}_\text{\normalsize{part 1}}
\ + \underbrace{\sum_{\substack{q: s'_q\notin\text{ any} \\ \text{block}}} 
16\times\sub{segineq}}_\text{\normalsize{part 2}} 
+ \underbrace{\sum_{\text{blocks $\blk_{qu}$}}\sum_{\ell=q}^u
  7\times\sub[\ell]{segineq}}_\text{\normalsize{part 3}}
+ \underbrace{\sum_{\text{blocks $\blk_{qu}$}}
9\times\Bigl(\sub{fBineq}+\sum_{\ell=q+1}^u\sub[\ell]{BBineq}\Bigr)}_\text{\normalsize{part 4}}.  
\label{cmpineq2}
\end{equation}
As before, the LHS of \eqref{cmpineq2} is 0, and we analyze the contribution  
from the $f^*_i, f_i, c^*_j, c_j$ terms to the RHS. Many of the terms are similar to those 
that appear in \eqref{cmpineq1}, so we state these without much elaboration.

Facilities $i\notin Z$ contribute at most $\sum_{i\in\Pc_c(Z)}32f^*_i$.
Consider $i\in Z$. If $o_i=o'_r$, then note that $s'_r$ does not belong to a block, and we
gather at most $(21+16)(f^*_i-f_i)+16\cdot 1.5(f_i+f^*_i)\leq 61f^*_i-13f_i$ from
parts 1 and 2.
Suppose $\sg(o_i)\in Z$. If $\sg(o_i)\in Z\sm A,\ s_i\notin B$, we gather $42f^*_i$ from
part 1. Let $s_i\in Z_q$. If $s'_q$ does not belong to any block then we pick up
$16(f^*_i-f_i)$ from part 2. Otherwise, note that $s'_q$ must be the start of a block, 
therefore, we pick up $7(f^*_i-f_i)+9\cdot 2f^*_i$ from parts 3 and 4. So the overall 
contribution is at most $67f^*_i-7f_i$.
Next, suppose $s_i\sib\sg(o_i)$. Let $s_i=s'_\ell$. We pick up $42f^*_i$ from part 1, 
$7\bigl[\bigl(1+\tfrac{n^*_{o'_{\ell-1}}}{w_i}\bigr)f^*_i
-\bigl(1-\tfrac{n^*_{o'_{\ell-1}}}{w_i}\bigr)f_i\bigr]$ from part 3, and
$9\bigl[\bigl(2+\frac{n^*_{o'_{\ell-1}}}{w_i}\bigr)f^*_i
-\bigl(2-\frac{n^*_{o'_{\ell-1}}}{w_i}\bigr)f_i\bigr]$ from part 4.
This amounts to at most $91f^*_i-f_i$ total contribution.
Suppose $\sg(o_i)\in A$. We gather $21\bigl(2f^*_i-(1-\al_i)f_i\bigr)$ from part 1.
If $s_i\notin B$, we gather $16\bigl[f^*_i-f_i+\frac{n^*_{o_i}}{w_i}(f_i+f^*_i)\bigr]$
from parts 2, 3, and 4. If $s_i\in B$, then $s_i=s'_u$ is the end of some block, and we  
gather $16\cdot\frac{n^*_{o_i}}{w_i}(f_i+f^*_i)+
7\bigl[\bigl(1+\tfrac{n^*_{o'_{u-1}}}{w_i}\bigr)f^*_i
-\bigl(1-\tfrac{n^*_{o'_{u-1}}}{w_i}\bigr)f_i\bigr]+9(f^*_i-f_i)$ from parts 2, 3, and 4. 
So the total contribution in both cases is at most
$92.5f^*_i-\bigl(5.5-\frac{2n^*_{o_i}}{w_i}\bigr)f_i\leq 92.5f^*_i-2.5f_i$.
If $\sg(o_i)\in C\sm A$ and $s_i\in B$, then $s_i=s'_u$ is the end of some block. 
We gather $21(2f^*_i-(1-\al_i)f_i)$ from part 1, and 
$7\bigl[\bigl(1+\tfrac{n^*_{o'_{u-1}}}{w_i}\bigr)f^*_i
-\bigl(1-\tfrac{n^*_{o'_{u-1}}}{w_i}\bigr)f_i\bigr]+9(f^*_i-f_i)$ from parts 3 and 4,
which amounts to at most $74.5f^*_i-5.5f_i$ total contribution.
Finally, if $\sg(o_i)\in Z\sm C$, $s_i\in B$ and $s_i\nsib\sg(o_i)$, we gather
$21(f^*_i-f_i)$ from part 1. If $s_i$ does not belong to a block, we gather 
$16\bigl[f^*_i-f_i+\frac{n^*_{\cent(s_i)}}{w_i}(f_i+f^*_i)\bigr]$ from part 2. If $s_i$
belongs to a block, then notice that it can only be the start $s'_q$ of the block. 
So $\sg(o_i)\in S_2$ and we gather 
$7(f^*_i-f_i)+9\cdot 2f^*_i+16\cdot\frac{n^*_{o'_{q-1}}}{w_i}(f_i+f^*_i)$ from parts 3 and
4. Accounting for both cases, we gather at most $70f^*_i-4f_i$.

Now consider the $\{c^*_j, c_j\}$-contribution.
This includes the terms $\sum_{j\in D((S'\cap S_3)\cup H(Z\cap S))} 32c^*_j$ and
\linebreak $\sum_{j\in D(s'_0)} 42c^*_j$.
We bound the remaining contribution.
Consider a client $j\in D^*(o'_q)$. Let $i=\hi_q$. Its (remaining) part-1 contribution is
$9(c_j+c^*_j)$ if $\sg(o_i)\in A$ or $s_i\in B,\ s_i\nsib\sg(o_i)$, and 0 otherwise.
If $s'_q$ is not in any block, we pick up at most $16(c^*_j-c_j)+16c^*_j$ from part 2. 
If $s'_q$ is the start of a block, we pick up $7(c^*_j-c_j)$ from part 3, and note that
$\sg(o_i)\notin A,\ s_i\notin B$. 
If $s'_q$ is an intermediate $S'$-location of a block, then $s_i=s'_q$ and
$s_i\sib\sg(o_i)$. We pick up $7(c^*_j-c_j)$ from part 3, and
$9\bigl(\frac{8}{3}c^*_j+\frac{2}{3}c_j\bigr)$ from part 4. So the total (remaining)
contribution is at most $41c^*_j-c_j$ in all cases.
Now consider $j\in D^*(o_i)$, where $o_i\notin O''$. We have $\sg(o_i)\in S_2$. Let
$o_i\in Z_q$. We accrue $14(c_j+c^*_j)$ from part 1 if $s_i\in B$ and 0 otherwise.
If $s'_q$ does not belong to a block, we gather at most $\frac{80}{3}(c_j+c^*_j)$ from
part 2; if $s'_q$ belongs to a block, it must be the start of the block, and we gather at
most $\frac{35}{3}(c_j+c^*_j)$ from part 3. So in total, we accrue at most
$\frac{122}{3}(c_j+c^*_j)$.

Thus, \eqref{cmpineq2} leads to inequality \eqref{pathineq2}.
\end{proof}

\begin{lemma} \label{cyclem}
Let $Z\in\C$ and $r=|S'|>0$. Define $s'_q=s'_{q\bmod r},\ o'_q=o'_{q\bmod r}$ for all
$q$. Then,
\begin{equation}
\begin{split}
0\ \leq & \sum_{i\in\Pc_c(Z)}32f^*_i+\sum_{i\in Z}(92.5f^*_i-f_i)
+\sum_{j\in D((S'\cap S_3)\cup H(Z\cap S))} 32c^*_j \\
& +\sum_{j\in D^*(Z\cap O'')}(41c^*_j-c_j)
+\sum_{j\in D^*((Z\cap O)\sm O'')}\frac{122}{3}(c_j+c^*_j).
\end{split}
\label{cycineq}
\end{equation}
\end{lemma}

\begin{proof}
The arguments are almost identical to those in the proofs of Lemmas~\ref{pathlem1}
and~\ref{pathlem2}. The only change is that we no longer have inequality \eqref{ineq1} or
\eqref{ineq2}. Instead, we start with the inequality 
$0\leq\sum_{i\in Z}\bigl(-f_i+f^*_i+w_i\cdot c(o_i,\sg(o_i))\bigr)$, and we bound
$w_i\cdot c(o_i,\sg(o_i))$ suitably, as in Lemma~\ref{pathlem1} or
Lemma~\ref{pathlem2}, to obtain \eqref{cineq1} and \eqref{cineq2} that are analogous 
to \eqref{ineq1} and \eqref{ineq2} respectively.
\begin{align}
& \begin{split}
0\ \leq\ & 
\sum_{\substack{i\in Z: \sg(o_i)\in Z\sm A \\ s_i\notin B}} \negthickspace\negthickspace 2f^*_i
+\sum_{\substack{i\in Z: \sg(o_i)\in A\text{ or} \\ \sg(o_i)\in C\sm A, s_i\in B}}
\negthickspace\negthickspace\bigl(2f^*_i-(1-\al_i)f_i\bigr)
+\sum_{\substack{i\in Z: \sg(o_i)\in Z\sm C \\ s_i\in B}}\negthickspace\negthickspace(f^*_i-f_i) \\
& +\sum_{\substack{i\in Z: \sg(o_i)\in A \\ \text{or }s_i\in B}}\sum_{j\in D^*(o_i)}
\negthickspace\negthickspace \tfrac{2}{3}(c_j+c^*_j). 
\end{split}
\label{cineq1} \\[1ex]
& \begin{split}
0\ \leq\ & 
\sum_{\substack{i\in Z: s_i\sib\sg(o_i) \text{ or} \\ \sg(o_i)\in Z\sm A, s_i\notin B}}  
\negthickspace\negthickspace 2f^*_i
+\sum_{\substack{i\in Z: \sg(o_i)\in A\text{ or} \\ \sg(o_i)\in C\sm A, s_i\in B}}
\negthickspace\negthickspace\bigl(2f^*_i-(1-\al_i)f_i\bigr)
+\sum_{\substack{i\in Z: \sg(o_i)\in Z\sm C \\ s_i\in B, s_i\nsib\sg(o_i)}}
\negthickspace\negthickspace(f^*_i-f_i) \\
& +\sum_{\substack{i\in Z: \sg(o_i)\in A\text{ or} \\ s_i\in B, s_i\nsib\sg(o_i)}}
\sum_{j\in D^*(o_i)}\tfrac{2}{3}(c_j+c^*_j).
\end{split}
\label{cineq2}
\end{align}
The rest of the proof proceeds as in Lemmas~\ref{pathlem1} and~\ref{pathlem2}: we
substitute \eqref{cineq1} for \eqref{ineq1} in the proof of Lemma~\ref{pathlem1}, and
substitute \eqref{cineq2} for \eqref{ineq2} in the proof of Lemma~\ref{pathlem2}. It is   
not hard to see then that we obtain inequality \eqref{cycineq}.
\end{proof}

\begin{theorem} \label{1swthm}
The cost of a local optimum using 1-swaps is at most $O(1)$ times the optimum solution
cost. 
\end{theorem}

\begin{proof}
Let $F$ and $C$ denote respectively the total movement- and assignment- cost of the local
optimum.
For a set $A\sse S$, let $\cent(A)$ denote $\bigcup_{s\in A}\{\cent(s)\}$. 
Summing \eqref{1sws2} for all $s\in S_2$ and simplifying, we obtain that
\begin{equation}
\sum_{j\in D^*(\cent(S_2))}c_j \leq \sum_{j\in D^*(\cent(S_2))}5c^*_j+\sum_{j\in D(S_2)}6c^*_j. 
\label{thmineq1}
\end{equation}
Summing \eqref{pathineq2} for every path $Z\in\Pc$, and \eqref{pcineq1} or \eqref{cycineq}
for every cycle $Z\in\C$, yields the following.
\begin{equation}
\begin{split}
F+\sum_{j\in D^*(O\sm\cent(S_2))}\negthickspace\negthickspace c_j\ \leq\
& 124.5\cdot F^*+\sum_{j\in D(S_0)}74c^*_j+\sum_{j\in D(S_1\cap S_3)}32c^*_j
+\sum_{j\in D^*(O\sm\cent(S_2))}41c^*_j \\
& +\sum_{j\in D^*(\cent(S_2))}\tfrac{122}{3}(c_j+c^*_j).
\end{split}
\label{thmineq2}
\end{equation}
Adding $\frac{125}{3}\times\eqref{thmineq1}$ to \eqref{thmineq2}, we get
that 
\begin{equation*}
\begin{split}
F+C & \leq 124.5\cdot F^*+\sum_{j\in D(S_0)}\negthickspace\negthickspace 74c^*_j
+\sum_{j\in D(S_1\cap S_3)}\negthickspace\negthickspace\negthickspace 32c^*_j
+\sum_{j\in D(S_2)}\negthickspace\negthickspace 250c^*_j
+\sum_{j\in D^*(O\sm\cent(S_2))}\negthickspace\negthickspace\negthickspace\negthickspace\negthickspace\negthickspace\negthickspace
 41c^*_j
+\sum_{j\in D^*(\cent(S_2))}\negthickspace\negthickspace\negthickspace\negthickspace\negthickspace\negthickspace 
249c^*_j \\
& \leq 124.5\cdot F^*+499\cdot C^*.
\end{split}
\end{equation*}
\end{proof}

\section{Bad locality gap with arbitrary facility-movement costs} \label{locgap} 
In this section, we present an example that shows that if the facility-movement costs and
the client-assignment costs come from different (unrelated) metrics 
then the $\lz$-swap local-search algorithm has an unbounded locality gap;
that is, the cost of a local optimum may be arbitrarily large compared to optimal cost. 

We first show a simple example for a single swap case, which we will later generalize for
$\lz$ swaps. Suppose we have two clients $j_0, j_1$ and two facilities 
$i_0, i_1$. Some distances between these clients and facilities are shown in the
Fig.~\ref{fig:locgap}(a); all other distances are obtained by taking the metric
completion. Note that in this example, in order to have a bounded movement 
cost for facilities, the only option is to have one of $i_0, j_1$ as a final location of
facility $i_0$ and one of $i_1, j_0$ as a final location of facility $i_1$. 

As can be seen from the figure, the solution $O=\{i_0,i_1\}$ has total cost $2$ (the
movement cost is $0$ and the client-assignment cost is $2$). Now consider the solution   
$S=\{j_0, j_1\}$ which has a total cost of $2D$ (the movement cost is $2D$ and the
client-assignment cost is $0$). This is a local optimum since if we swap out $j_0$, then
we have to swap in $i_1$ to have a bounded movement cost, which leads $j_0$ having
assignment cost of $\infty$. 
By symmetry, there is no improving move for solution $S$, and the locality gap is $D$. 

\begin{figure}[ht!]
\centerline{\resizebox{!}{2.5in}{\input{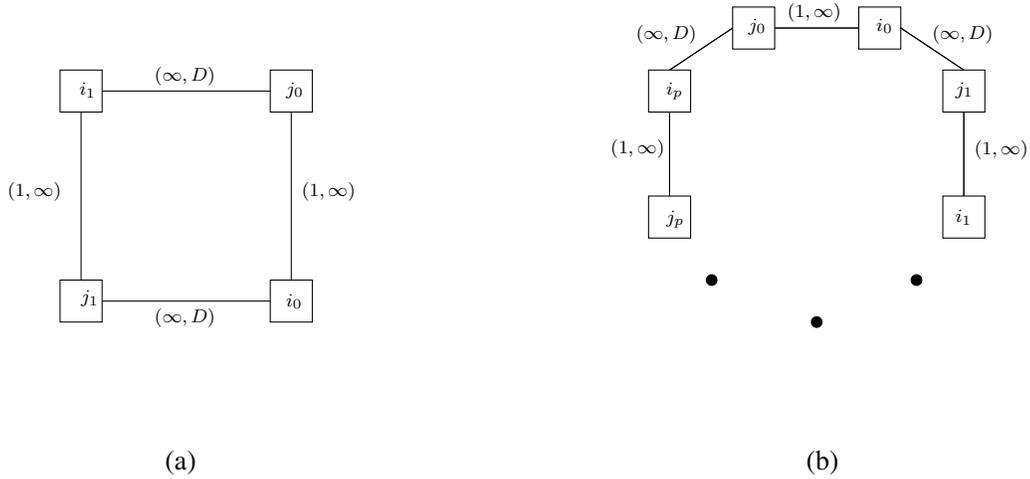}}}
\caption{Examples showing large locality gap for the cases where local search allows (a)
  single swaps (b) at most $\lz$ simultaneous swaps. The label $(a,b)$ of an edge gives
  client-assignment cost $a$ and the movement cost $b$ of a facility along that edge.}   
\label{fig:locgap}
\end{figure}

Now consider the example shown in Fig.~\ref{fig:locgap}(b) for local-search with $\lz$
simultaneous swaps. Suppose we have facility set $\{i_0,i_1,\ldots,i_{p}\}$ and client set 
$\{j_0,j_1,\ldots,j_{p}\}$. The global optimum $O=\{i_0,i_1,\cdots,i_{p}\}$ has 
total cost $p+1$ (facility movement cost is $0$ and client-assignment cost is
$(p+1)\cdot 1$) while $S=\{j_0,j_1,\ldots,j_{p}\}$ is a local optimum whose total cost is
$(p+1)\cdot D$ (facility movement cost is $(p+1)\cdot D$ and client-assignment cost is
$0$). Consider any move $swap(X,Y)$. Note that $j_k\in X$ iff $i_{k-1}\in Y$
(where indices are $\ \bmod (p+1)$) to ensure bounded movement cost.
Let $k$ be such that $j_k\in X$ and $j_{k+1}\notin X$. Then, $j_k$ has an assignment cost
of $\infty$ in the solution $(S\sm X)\cup Y$. Hence, $S$ is a local optimum.

\end{document}